\documentclass[12pt]{article}
\usepackage{a4wide}
\usepackage[utf8]{inputenc}
\usepackage[T1]{fontenc}
\usepackage[english]{babel}
\usepackage{lmodern,eurosym}
\usepackage{amsmath}
\usepackage{mathrsfs}
\usepackage{amsfonts}
\usepackage{amssymb}
\usepackage{graphicx}
\usepackage{verbatim}
\usepackage{tikz}
\usetikzlibrary{arrows,patterns,topaths}
\usetikzlibrary{positioning}
\tikzstyle{call} = [->, line width=1mm, lightgray]
\tikzstyle{callsequencelabel} = [black, font=\footnotesize]
\usepackage{array}
\usepackage{float}
\usepackage[pdfusetitle,hidelinks]{hyperref}

\usepackage[T1,safe]{tipa}
\usepackage{listings}
\usepackage{amssymb}
\usepackage{color}
\definecolor{keywordcolor}{rgb}{0.7, 0.1, 0.1}   
\definecolor{tacticcolor}{rgb}{0.0, 0.1, 0.6}    
\definecolor{commentcolor}{rgb}{0.4, 0.4, 0.4}   
\definecolor{symbolcolor}{rgb}{0.0, 0.1, 0.6}    
\definecolor{sortcolor}{rgb}{0.1, 0.5, 0.1}      
\definecolor{attributecolor}{rgb}{0.7, 0.1, 0.1} 

\newcommand{\lean}[1]{%
  \href{https://m4lvin.github.io/Gossip-in-Lean/docs/find/?pattern=Error.#1\#doc}{\texttt{\detokenize{#1}}}%
}

\usepackage[thmmarks]{ntheorem}
\theorembodyfont{\rmfamily}
\theoremsymbol{\ensuremath{\dashv}}
\usepackage{newproof}

\newtheorem{theorem}{Theorem}
\newtheorem{proposition}[theorem]{Proposition}
\newtheorem{lemma}[theorem]{Lemma}
\newtheorem{corollary}[theorem]{Corollary}
\newtheorem{conjecture}[theorem]{Conjecture}

\newtheorem{definition}[theorem]{Definition}
\newtheorem{example}[theorem]{Example}
\newtheorem{observation}[theorem]{Observation}

\newcommand{\LNS}{\mathsf{LNS}}

\newcommand{\ANY}{\mathsf{ANY}}
\newcommand{\prot}{\mathsf{P}}

\renewcommand{\phi}{\varphi}

\newcommand{\weg}[1]{}

\newcommand{\Exp}{\mathit{Exp}}
\newcommand{\cExp}{\mathit{Exp}^\mathit{cor}}

\newcommand{\eq}{\leftrightarrow}
\newcommand{\Eq}{\Leftrightarrow}
\newcommand{\imp}{\rightarrow}
\newcommand{\Imp}{\Rightarrow}

\newcommand{\Pmi}{\Leftarrow}
\newcommand{\et}{\wedge}
\newcommand{\vel}{\vee}
\newcommand{\Et}{\bigwedge}

\newcommand{\M}{\hat{K}}
\renewcommand{\phi}{\varphi}
\newcommand{\union}{\cup}

\newcommand{\inter}{\cap}

\newcommand{\bigO}{\mathcal{O}}

\newcommand{\Kw}{\mathit{Kw}}

\newcommand{\ov}[1]{\overline{#1}}
\newcommand{\un}[1]{\underline{#1}}

\usepackage{soul}

\usepackage{bm}

\newcommand{\last}{{\mathit{last}}}
\newcommand{\full}{{\mathit{full}}}

\newcommand{\power}{{\mathcal{P}}}

\newcommand{\Kv}{\mathit{Kv}}
\newcommand{\fv}{\mathsf{v}^\sim}

\newcommand{\initsecret}{S}
\newcommand{\initsecreto}{I}
\newcommand{\initsecretp}{T}

\newcommand{\mmodel}{{\mathcal{M}}}
\newcommand{\cor}{{\mathit{cor}}}

\begin{document}

\title{Self-Correcting Gossip Protocols}

\author{Giorgio Cignarale \and Hans van Ditmarsch\thanks{Author affiliations of Giorgio, Stephan: TU Wien, Austria; Hans: CNRS, IRIT, University of Toulouse, France; Malvin: ILLC, University of Amsterdam, Netherlands; Hugo: TU Berlin, Germany; Vaishnavi: IIT Delhi, India. Hans van Ditmarsch, {\tt hansvanditmarsch@gmail.com}, is corresponding author. We acknowledge substantial contributions from Roman Kuznets and Ulrich Schmid to this work.} \and Stephan Felber \and Malvin Gattinger \and Hugo Rincon Galeana \and Vaishnavi Sundararajan}
\date{}

\maketitle

\begin{abstract}
We investigate self-correcting gossip protocols with errors. In distributed computing, protocols with errors have been widely investigated in temporal epistemic logics. Instead, we propose a dynamic epistemic logic. 
We show how to correct transmission errors due to faulty messages without a central authority coordinating protocol execution, how this affects optimality, and how this compares to bounded memory and full information protocols. 
\end{abstract}


\section{Introduction}\label{sec.introduction}

In \emph{gossip protocols} \cite{tijdeman:1971,BAKER1972191,west82a,hedetniemietal:1988,kermarrecetal:2007,attamahetal.ecai:2014,AptW18,hvdetal.dynamicgossip:2019,CooperHMMR19,logicofgossiping:2020}, given a set of $n$ \emph{agents}, each agent knows a single \emph{secret} (we can think of this as the identity of that agent, or as all the agent knows); when agents call each other, they exchange all secrets they know; and the goal is for all agents to become \emph{experts}, that is, to know all secrets. Here, we assume that all agents can call each other (are all neighbours), so that there are no network constraints, and we also assume that messages are always received. A stronger epistemic goal is that all agents are \emph{super experts}, that is, all agents know that all agents know all secrets. If only secrets are exchanged, even higher-order epistemic goals are unreachable \cite{hvdetal.lucky:2024}.

The usual assumption in epistemic gossip protocols is that message transmission is correct, and that agents behave correctly. Under these assumptions, and in the absence of network constraints, $2n-4$ calls are optimal to reach the goal that all become experts \cite{tijdeman:1971}, and $n-2 + \binom{n}{2}$ calls are optimal to reach the goal that all become super experts \cite{hvdetal.lucky:2024}. What epistemic goals are still reachable and how is optimality affected when transmission errors may occur during a call, and when agents may behave incorrectly? In distributed computing, bounds on transmission errors are investigated in population protocols \cite{AspnesR07}, the epistemic consequences of faulty behaviour is investigated in \cite{KuznetsP0F19,abs-2106-11499}, and self-correction (self-stabilization) is investigated in \cite{DaliotD05,DolevFPSSL14,dolev:2000}. Faulty agents in epistemic gossip protocols have been investigated in \cite{line:2018,BergG20,tinafurer:2023}. In \cite{BergG20} the authors investigate how the presence of unreliable agents affects whether the gossip protocol can terminate successfully (that is, whether all agents get to know all secrets). They investigate this for dynamic gossip \cite{hvdetal.dynamicgossip:2019} wherein agents not only exchange secrets but also numbers (relaxing network constraints), and for various epistemic gossip protocols including the protocol called $\LNS$ wherein one can only call a neighbour if one does not know her secret. They not only require that all reliable agents get to know the secrets of all reliable agents, but also that the reliable agents get to know who the unreliable agents are --- the latter is an uncommon requirement in distributed computing. Publication \cite{line:2018} is a precursor of \cite{BergG20}, and \cite{tinafurer:2023} is a follow-up proposing yet other variations of $\LNS$ for unreliable agents.

In this work we investigate in depth what happens when at most one transmission error may occur during protocol execution. We assume synchronous communication where agents are only aware of the calls involving them but are aware of a global clock. We show how a transmission error may cause false beliefs, how to correct this without a central authority coordinating protocol execution and how then to obtain epistemic protocol goals. We do not investigate optimality except by examples demonstrating that lower bounds must be higher. We show how our approach compares to (weaker) bounded memory protocols \cite{DH08} and to (stronger) full information protocols \cite{MosesT88}. We later wish to pursue how the same or a similar formalization can be used for self-correcting gossip protocols with at most $f$ faulty messages or with at most $f$ faulty (Byzantine) agents, and for an asynchronous setting. 


\section{Gossip with at most one transmission error} \label{section.onetransmission}

\subsection{Structures, syntax and semantics} \label{section.syntaxsemantics}

\paragraph*{Secret distributions}


Let a set of $n$ \emph{agents} $A = \{a_1,\dots,a_n\}$ be given. We typically assume few agents in which case they are named $a,b,c,\dots$ instead. 
Given the agents $A$, the set of \emph{secrets} is the product $A \times \{0,1\}$, where for $(a,1)$ we write $a$ and for $(a,0)$ we write $\ov{a}$. The overloaded use of agent names as secret values is disambiguated by context. If $B \subseteq A \times \{0,1\}$ is a subset of secrets, we define the \emph{swap} $\pm{a}$ of the values for the secret of agent $a$ in that set as $B^{\pm{a}}:= B[a/\ov{a},\ov{a}/a]$ (if $a$ is in $B$ replace it by $\ov{a}$ and simultaneously if $\ov{a}$ is in $B$ replace it by $a$). Note that $(B^{\pm{a}})^{\pm{a}}=B$, and if $B$ has no information on $a$ then $B^{\pm{a}}=B$.

\begin{definition}[Secret distribution]
A \emph{secret distribution} is a function \[ S: A \imp \power(A \times \{0,1\}) \] For $S(a)$ we write $S_a$ (the \emph{holding} of $a$). In an \emph{initial secret distribution} $\initsecret$, for all $a\in A$, $\initsecret_a =\{a\}$ or $\initsecret_a =\{\ov{a}\}$. The set of initial secret distributions is $\bm{I}$. In {\em {\bf the} initial secret distribution} denoted $\initsecreto$, for all $a\in A$, $\initsecreto_a =\{a\}$.
\end{definition}
A secret distribution lists for each agent what secrets that agent holds and what their values are. We let $S_a^{\inter{b}}$ denote $S_a \inter \{b,\ov{b}\}$. This is the set of values that agent $a$ holds for secret $b$.

It is convenient to have an abbreviated notation for secret distributions: a lexicographically ordered list of $n$ holdings of secrets $S_a$, $S_b$, \dots that are separated by vertical bars $|$, where each $S_a$ is also written as an ordered list but without separation symbols (so $\{a,b,c\}$ becomes $abc$), and where holding two values $b$ and $\ov{b}$ for the same secret is represented as a {\em conflicting value} $\un{b}$.  We can thus view the secrets $S_a$ held by agent $a$ as an \emph{annotated} subset $B$ of $A$. An example is secret distribution $S = ab\un{c}|ab\ov{c}|abc|d$, wherein $S_a = \{a,b,c,\ov{c}\}$, etcetera. We often assume the initial secret distribution $\initsecreto$, for example $a|b|c|d$.

In error-free gossip, at any stage an agent $a$ holds some subset $B \subseteq A$ of all secrets, including its own, and secret distributions are $n$-tuples $B_1|\dots|B_n$ with $B_1,\dots,B_n\subseteq A$. By assuming any such secret $b \in A$ held by $a$ to be a secret value $(b,1)$ in our setting with errors, such error-free secret distributions are now special cases, such as $a|b|c|d$ above.

\paragraph*{Call and call sequence}

In error-free gossip a call is a pair $(a,b)$ where $a \neq b \in A$, denoted $ab$, and which means that $a$ calls $b$, and wherein the agents exchange all their (values of) secrets. We say that $a$ and $b$ are \emph{involved} in the call. In gossip with errors, $ab$ is a \emph{correct call} and we also consider a \emph{faulty call} $a^cb$ where in the call from $a$ to $b$ there is a transmission error made in the secret $c$ held by $a$, so that $b$ receives the other value for secret $c$. Similarly we define $ab^c$ where in the call from $a$ to $b$ there is a transmission error made in the secret $c$ known by $b$. Both a correct call and a faulty call are now a \emph{call}. An arbitrary call is denoted $\kappa$ (for `singleton $\kappa$all sequence'), and an arbitrary call between $a$ and $b$, in either direction, is denoted $ab^\kappa$.

\begin{definition}[Call sequence]
A \emph{call sequence} is a finite sequence of calls containing at most one faulty call. Call sequences are denoted $\sigma,\tau$, and $\epsilon$ is the empty call sequence. We let $\sigma\sqsubseteq\tau$ mean that call sequence $\sigma$ is a \emph{prefix} of call sequence $\tau$, and $\sigma.\tau$ is \emph{concatenation} of call sequences.
\end{definition}

\paragraph{Syntax}

The set of \emph{atoms} (\emph{propositional variables}) $P$ is the product $(A \times \{0, 1\}) \times A$ where, analogously to the convention for secrets, for $((b,0), a)$ we write $\ov{b}_a$ and for $((b,1),a)$ we write $b_a$. 

\begin{definition}[Logical language]
The logical language consists of \emph{formulas} 
\[ \phi  ::= b_a \mid \ov{b}_a \mid \neg \phi \mid \phi \et \phi \mid K_a \phi \] 
\end{definition}
Other propositional connectives are defined by notational abbreviation, $K_a \phi$ stands for `agent $a$ knows that $\phi$', $\M_a \phi := \neg K_a \neg \phi$ stands for `agent $a$ considers $\phi$ possible', $E_B \phi := \Et_{a \in B} K_a\phi$ means that everyone in $B$ knows $\phi$ ($\phi$ is \emph{mutual knowledge} among the agents in $B$). In error-free gossip $b_a$ means that agent $a$ \emph{knows} or \emph{holds} the secret of agent $b$. In our setting with errors $b_a$ means that $a$ \emph{holds value} $b$ of the secret of $b$ whereas $\ov{b}_a$ means that $a$ holds value $\ov{b}$ of secret $b$. We finally define $\Kv_a b := K_a b_b \vel K_a \ov{b}_b$, for `agent $a$ knows (the value of) the secret of $b$' or `$a$ knows secret $b$'. (The notation for knowing value is reminiscent of notation $\Kw_a p$ for `knowing whether $p$', defined as $K_a p \vel K_a \neg p$.) After introducing the semantics we will show how all these epistemic readings relate.

Definitions~\ref{def.semanticscall}, \ref{def.observationrelation} and \ref{def.semanticsformulas} of respectively the semantics of a call, the observation relation, and the semantics of formulas, are defined by simultaneous induction.

\paragraph*{Semantics of a call}
In a call the agents exchange all the secrets they know. The semantics of a call do not depend on the direction of the call, so the semantics of calls $ab$ and $ba$ are the same.  A call from agent $a$ to agent $b$ affects the secret distribution. In standard gossip, if $a$ holds $X \subseteq A$ and $b$ holds $Y \subseteq A$, then after the call agents $a$ and $b$ both hold $X \union Y$. For gossip with errors the semantics of a call are slightly more complex: given set of values of secrets $X$ and $Y$ we still take their union, except when the agent already knew the correct value of a secret $d$ before the call, and also when the agent knows the correct value of a secret $d$ after the call. In those cases, if the union of $X$ and $Y$ contains conflicting values $d$ and $\ov{d}$ for $d$, we then remove the value that is known to be incorrect. 

\begin{definition}[Semantics of call] \label{def.semanticscall}
Given an initial secret distribution $\initsecret$ and a call sequence $\sigma$, the secret distribution $\initsecret[\sigma]$ is defined by induction on $\sigma$. For the basis, $\initsecret[\epsilon] := \initsecret$. For the induction, given call sequence $\sigma$ and $a,b,c,d \in A$ with $a \neq b$ and $d \neq a,b$:
\[\begin{array}{lclcllcl}
\initsecret[\sigma.ab]_a &=& \initsecret[\sigma.ba]_a \ = \  \initsecret[\sigma.a^cb]_a \ = \ \initsecret[\sigma.ba^c]_a &=& (\initsecret[\sigma]_a \union \initsecret[\sigma]_b {\setminus} *) {\setminus} {**} \\
\initsecret[\sigma.ab^c]_a &=& \initsecret[\sigma.b^ca]_a &=& (\initsecret[\sigma]_a \union \initsecret[\sigma]_b^{\pm{c}} {\setminus} *) {\setminus} {**}  \\
\initsecret[\sigma.bd]_a &=& \initsecret[\sigma.b^cd]_a \ = \ \initsecret[\sigma.bd^c]_a &=& \initsecret[\sigma]_a
\end{array}\]
The set $*$ consists of \emph{known incorrect values of secrets} by agent $a$ after call sequence $\sigma$. The set $*$ is defined as: $* = \{ d \mid  \initsecret,\sigma \models K_a \ov{d}_d \} \union \{ \ov{d} \mid  \initsecret,\sigma \models K_a d_d \}$. The set $**$ is defined as follows. For $ab^\kappa = ab,ba,a^cb,ba^c$ the set $**$ consists of all secret values $d$ such that $\initsecretp,\tau \models \ov{d}_d$ for all initial secret distributions $\initsecretp$ and call sequences $\tau$ with $(\initsecret,\sigma)\sim_a (\initsecretp,\tau)$ and $\initsecret[\sigma]_b = \initsecretp[\tau]_b$, and all secret values $\ov{d}$ such that $\initsecretp,\tau \models d_d$ for all $\initsecretp,\tau$ with $(\initsecret,\sigma)\sim_a (\initsecretp,\tau)$ and $\initsecret[\sigma]_b = \initsecretp[\tau]_b$. For $ab^\kappa = ab^c,b^ca$ we replace $\initsecretp[\tau]_b$ by $\initsecretp[\tau]_b^{\pm{c}}$.
\end{definition}

In the definition above, $\models$ is the satisfaction relation of Definition~\ref{def.semanticsformulas}, below. 

The set $**$ is defined in a roundabout way for technical reasons: as the $**$-discarded values define the semantics of a call $ab^\kappa$, we cannot use the more intuitive formulation that set $**$ consists of the known incorrect values after sequence $\sigma.ab^\kappa$, as that would be circular. However, we can still `think' of agent $a$ first receiving the holding of agent $b$, then concluding she now knows a secret $d$ as a result of that, and finally $**$-discarding the value of that secret she now knows to be incorrect. All this is combined in the semantics of a single call.

The sets $*$ and $**$ can both be considered forms of self-correction. Intuitively, the set $\ast$ consists of values of secrets contributed by agent $b$ that agent $a$ refuses to incorporate because she already knows the correct values of those secrets (in case she had self-corrected for those secrets, then she already did so in the past). Whereas intuitively the set $**$ consists of the values of secrets that agent $a$ got to know during the call involving $a$ and $b$ and for which she had a conflict before that call, and that she then discards. Set $*$ is self-correction in the form of persistence of correct values; `refusal' to accept incorrect values. Whereas $**$ is self-correction in the form of `discarding' values held before the call.

As there is at most one transmission error, sets $*$ and $**$ can only be about the same unique secret $d$ in a given call sequence and cannot both be non-empty. But the semantics is presented in a general form that is also suitable for multiple errors.

From here on, $S$ and $T$ always denote \emph{initial} secret distributions (while still being declared as such) while non-initial secret distributions always take shape $S[\sigma]$, $T[\tau]$, $I[\sigma]$, etcetera.

\paragraph*{Observation model}

The \emph{observation model} $\mmodel(\bm{I})$ is the Kripke model $(W,\sim,V)$ where \emph{domain} $W$ consists of \emph{gossip states} $(\initsecret,\sigma)$ for initial secret distributions $\initsecret \in \bm{I}$ and call sequences $\sigma$, where for each $a \in A$ \emph{observation relation} $\sim_a$ between gossip states is defined below, and where \emph{valuation} $V$ maps a gossip state to a secret distribution such that $V(\initsecret,\sigma)= \initsecret[\sigma]$.
\begin{definition}[Observation relation] \label{def.observationrelation} The {\em observation relation} is the equivalence closure of the following recursive clauses by call sequence length, where $S$ and $T$ are initial secret distributions, and where $a,b,c,d,e,f,g \in A$ with $a \neq b$ and $c,d,f,g \neq a$, and the clauses for the other direction of the call between $a$ and $b$, and for the other direction of the faulty call between $c$ and $d$, are the same. 
\[\begin{array}{lll}
(\initsecret,\epsilon) \sim_a (\initsecretp,\epsilon) & \text{iff} & \initsecret_a = \initsecretp_a \\
(\initsecret,\sigma.ab) \sim_a (\initsecretp,\tau.ab) & \text{iff} & (\initsecret,\sigma) \sim_a (\initsecretp,\tau) \text{ and } \initsecret[\sigma]_b = \initsecretp[\tau]_b \\
(\initsecret,\sigma.ab) \sim_a (\initsecretp,\tau.a^eb) & \text{iff} & (\initsecret,\sigma) \sim_a (\initsecretp,\tau) \text{ and } \initsecret[\sigma]_b = \initsecretp[\tau]_b\\
(\initsecret,\sigma.ab) \sim_a (\initsecretp,\tau.ab^e) & \text{iff} & (\initsecret,\sigma) \sim_a (\initsecretp,\tau) \text{ and } \initsecret[\sigma]_b = \initsecretp[\tau]_b^{\pm e} \\
(\initsecret,\sigma.cd) \sim_a (\initsecretp,\tau.fg) & \text{iff} & (\initsecret,\sigma) \sim_a (\initsecretp,\tau) \\ 
(\initsecret,\sigma.c^ed) \sim_a (\initsecretp,\tau.fg) & \text{iff} & (\initsecret,\sigma) \sim_a (\initsecretp,\tau)
\end{array}\]
\end{definition}
If $(\initsecret,\sigma) \sim_a (\initsecretp,\tau)$, then $|\sigma|=|\tau|$. We note that $\sim_a$ therefore defines a \emph{synchronous} observation relation \cite{kermarrecetal:2007,attamahetal.ecai:2014,AptW18,logicofgossiping:2020}: agents are only aware of calls involving them, but are still aware that a call took place if they were not involved in a call; calls are scheduled in `rounds' consisting of single calls thus defining a global clock.

In error-free gossip we only need to consider the initial secret distribution $\initsecreto$, we therefore omit that parameter, and then define: $\epsilon \sim_a \epsilon$, $\sigma.ab \sim_a \tau.ab$ iff $\sigma \sim_a \tau$ and $\sigma_b=\tau_b$, and $\sigma.cd \sim_a \tau.fg$ iff $\sigma \sim_a \tau$ ($\initsecreto[\sigma]_b$ is then denoted $\sigma_b$) \cite{DitmarschGR23,DitmarschEPRS17}. 

In the definition of the observation relation, the restriction to call sequences containing at most one transmission error is only implicit. Given $(\initsecret,\sigma) \sim_a (\initsecretp,\tau)$, in case $\tau$ contains a faulty call, the extension of $\tau$ with $a^eb$ or $ab^e$ is not a call sequence. So, turning the matter around, the clauses involving $\tau.a^eb$ and $\tau.ab^e$ in the definition imply that $\tau$ contains no faulty calls. This means that equivalence classes of $\sim_a$ indistinguishable gossip states for sequences $\sigma$ involving more and more calls, involve fewer and fewer initial secret distributions $\initsecret$, thus eventually resulting in knowledge of secrets, namely when only one secret distribution is considered possible by the agent. 

\paragraph*{Semantics of formulas}
The formulas of the logical language are interpreted in the observation model. 
\begin{definition}[Semantics] \label{def.semanticsformulas}
Given a gossip state $(\initsecret,\sigma)$, the semantics are defined by induction on formula structure.
\[\begin{array}{lll}
\initsecret,\sigma \models b_a & \text{ iff } & b \in \initsecret[\sigma]_a \\
\initsecret,\sigma \models \ov{b}_a & \text{ iff } & \ov{b}\in\initsecret[\sigma]_a \\
\initsecret,\sigma \models \neg\phi & \text{ iff } & \initsecret,\sigma \not\models \phi \\
\initsecret,\sigma \models \phi\et\psi & \text{ iff } & \initsecret,\sigma \models \phi \text{ and } \initsecret,\sigma \models \psi \\
\initsecret,\sigma \models K_a\phi & \text{ iff } & \initsecretp,\tau \models \phi \text{ for all } (\initsecretp,\tau) \sim_a (\initsecret,\sigma)
\end{array}\]
A formula $\phi$ is \emph{valid}, notation $\models\phi$, if $\initsecret,\sigma\models\phi$ for all initial secret distributions $\initsecret$ and call sequences $\sigma$.
\end{definition}
Given initial secret distribution $\initsecret$ and call sequence $\sigma$, \emph{$a$ holds the correct value of secret $b$} if $\initsecret_b\subseteq\initsecret[\sigma]_a^{\inter b}$ (inclusion, because $a$ may hold conflicting values), and \emph{$a$ knows the (correct) value of secret $b$} if $\initsecret,\sigma\models \Kv_a b$, which implies that $\initsecret_b=\initsecret[\sigma]_a^{\inter b}$ (Proposition~\ref{prop.kcorrectb} and its consequences, later). Knowing a secret implies the value is correct. A value is \emph{faulty} if it is not correct. 

There are a whole lot of epistemic or vaguely epistemic readings by now. Let us list them here all together in order to put it in proper perspective. We also attach an epistemic reading to combinations of certain values, as in our setting errors are rare.

\begin{itemize}
\item $b_a$ means that $a$ holds value $b$ of the secret of $b$; 
\item $\ov{b}_a$ means that $a$ holds value $\ov{b}$ of the secret of $b$.
\item $b_a \et \neg \ov{b}_a$ means that $a$ believes that the secret of $b$ is $b$;
\item $\neg b_a \et \ov{b}_a$ means that $a$ believes that the secret of $b$ is $\ov{b}$; 
\item $b_b \et b_a \et \neg \ov{b}_a$ means that $a$ correctly believes that the secret of $b$ is $b$;
\item $\ov{b}_b \et \neg b_a \et \ov{b}_a$ means that $a$ correctly believes that the secret of $b$ is $\ov{b}$; 
\item $\neg b_b \et b_a \et \neg \ov{b}_a$ means that $a$ incorrectly believes that the secret of $b$ is $b$;
\item $\neg \ov{b}_b \et \neg b_a \et \ov{b}_a$ means that $a$ incorrectly believes that the secret of $b$ is $\ov{b}$; 
\item $b_a \et \ov{b}_a$ means that $a$ has conflicting/inconsistent beliefs about the secret of $b$;
\item $\neg b_a \et \neg \ov{b}_a$ means that $a$ has no beliefs about the secret of $b$;
\item $\Kv_a b$ means that $a$ knows the secret of $b$ (that $a$ knows the value of the secret of $b$).
\end{itemize}
We recall that $\Kv_a b$ is defined as $K_a b_b \vel K_a \ov{b}_b$. 
We will show that $\Kv_a b \imp (b_b \et b_a \et \neg \ov{b}_a) \vel (\ov{b}_b \et \neg b_a \et \ov{b}_a)$. Knowledge implies correct belief. However, correct belief does not imply knowledge (of which the simplest example is that $\initsecreto,ab\models b_b \et b_a \et \neg \ov{b}_a$ whereas $\initsecreto,ab\not\models \Kv_a b$, see also Example~\ref{example.dadada}). In our setting, knowledge is  \emph{justified} correct belief. In the final section we summarily discuss the introduction of belief (quasi-)modalities for some of the above.

\paragraph*{The semantics are well-defined} \label{par:SemWellDef} On the one hand, observation relations $\sim_a$ for agent $a$ in the observation model are between gossip states $(\initsecret,\sigma)$ that are valued as secret distributions $\initsecret[\sigma]$ and such secret distributions $\initsecret[\sigma]$ depend on what agent $b$ knows and thus on $\models$ (whether $\initsecret,\sigma\models \Kv_b c$ for secrets $c$). On the other hand, the satisfaction relation $\models$ defines what agent $b$ knows as a function of the observation relation $\sim_b$, of which the definition depends on the holding $\initsecretp[\tau]_b$ of agent $b$ in secret distributions $\initsecretp[\tau]$. We should therefore show that the logical semantics are well-defined. 

Consider the following order $<$ on pairs of call sequences and formulas: \[ (\sigma,\phi) < (\tau,\psi) \ \text{ iff } \ (|\sigma| = |\tau| \text{ and } \phi \text{ is a strict subformula of } \psi) \text{ or } |\sigma| < |\tau|
\] Note that this is a well-founded partial order.

{\bf Satisfaction relation} With a clause $\initsecret,\sigma \models \phi$ we associate a pair $(\sigma,\phi)$. All definiens of the inductive clauses of the definition of $\models$ satisfy the order $<$: $(\sigma,\phi) < (\sigma,\neg\phi)$, $(\sigma,\phi) < (\sigma,\phi\et\psi)$, and $(\tau,\phi) < (\sigma,K_a\phi)$, where we observe that $|\tau|=|\sigma|$. For the basic clauses $\initsecret,\sigma \models b_a$ and $\initsecret,\sigma \models \ov{b}_a$ we need to use the inductive definition of $\initsecret[\sigma]$. Let us deal with $b_a$ where $\ov{b}_a$ is handled similarly. If $\sigma=\epsilon$, then we check whether $b=a$ and $\{a\}= \initsecret_a$ (a pair $(\epsilon,b_a)$ is at the bottom of the order). For the cases $\sigma = \tau.ac^\kappa$ we use that: $(\tau,b_a) < (\tau.ac, b_a),(\tau.a^dc, b_a)$ and $(\tau,b_c) < (\tau.ac, b_a), (\tau.a^dc, b_a)$ or even, when the precall is faulty about $b$, $(\tau,\ov{b}_c) < (\tau.ac^b, b_a)$, whereas if $\sigma = \tau.cd^\kappa$ we more straightforwardly have that $(\tau,b_a) < (\tau.cd^\kappa, b_a)$.

{\bf Observation relation} With a clause $(\initsecret,\sigma) \sim_a (\initsecretp,\tau)$ we associate pairs $(\sigma,\phi)$ and $(\tau,\phi)$ for any $\phi$. We thus get that $(\sigma, \phi), (\tau, \phi) < (\sigma.ab^\kappa,\phi), (\tau.ab^\kappa,\phi)$ because $|\tau|=|\sigma|< |\sigma.ab^\kappa| = |\tau.ab^\kappa|$. Similarly, $(\sigma, \phi), (\tau,\phi) < (\sigma.cd^\kappa,\phi), (\tau.ef^\kappa,\phi)$. This takes care of the $\sim_a$ part of the definiens. However, we also need to show that the identification of holdings in the definiens can be carried out. For that we use that $(\sigma, c_b) < (\sigma.ab^\kappa, c_b)$ and $(\sigma, \ov{c}_b) < (\sigma.ab^\kappa, {c}_b)$ for arbitrary secrets $c$ held by an agent $b$ in $\initsecret[\sigma]_b$, and we proceed similarly for $\initsecretp[\tau]_b$ and $\initsecretp[\tau]_b^{\pm{e}}$.

{\bf Secret distribution} With a clause $\initsecret[\sigma]_a$ (for any $a$) we associate a pair $(\sigma,\phi)$ for any $\phi$. We thus get that $(\sigma,\Kv_a d) < (\sigma.ab^\kappa, d_a), (\sigma.ab^\kappa, \ov{d}_a)$, because $|\sigma| < |\sigma.ab^\kappa|$, so that the test whether $\initsecret,\sigma\models \Kv_a d$, required to determine $\initsecret[\sigma.ab^\kappa]_a$, is well-defined. Finally we have that $(\sigma,d_a) < (\sigma.bc^\kappa, d_a)$ and $(\sigma,\ov{d}_a) < (\sigma.bc^\kappa, \ov{d}_a)$ for $b,c \neq a$.

\medskip

A {\bf Lean} implementation for this logical semantics is available at \url{https://m4lvin.github.io/Gossip-in-Lean/docs/Gossip/Error/Basic.html} and is also summarily described in the Appendix.

\subsection{Semantic results} \label{subsec:SemanticResults}

As the knowledge modality is interpreted on structures with equivalence relations, we get the usual validities and validity preservations for knowledge \cite{hvdetal.handbook:2015}:
\[\begin{array}{ll}
\models K_a \phi \imp \phi \ (\mathbf{T}) & \models K_a (\phi \imp \psi) \imp K_a \phi \imp K_a \psi \ (\mathbf{K}) \\
\models K_a \phi \imp K_a K_a \phi \ (\mathbf{4}) & \models \phi \text{ implies } \models K_a \phi \ (\mathbf{Nec}) \\
\models \neg K_a \phi \imp K_a \neg K_a \phi \ (\mathbf{5}) \qquad \ 
\end{array}\]

As a sanity check we show a number of elementary properties of the observation relation and the epistemic semantics.

First in line is that when two gossip states are indistinguishable for an agent, then the agent must hold the same values of secrets in both.

\begin{lemma} \label{lemma.simtolocal}
$(\initsecret,\sigma) \sim_a (\initsecretp,\tau)$ implies $\initsecret[\sigma]_a = \initsecretp[\tau]_a$.
\end{lemma}
\begin{proof}
The proof is by induction on the length of call sequence $\sigma$. 

If $\sigma=\epsilon$, then $\sigma=\tau = \epsilon$, and $(\initsecret,\epsilon) \sim_a (\initsecretp,\epsilon)$ implies that $\initsecret_a = \initsecretp_a$.

Let now $\sigma=\sigma'.\kappa$ where call $\kappa$ involves agents $a$ and $b$. We observe that $\tau$ must then have shape $\tau'.\kappa'$ where $\kappa'$ also involves $a$ and $b$ (and where the direction of the call is the same in $\kappa$ and $\kappa'$). In all such cases, by the definition of the observation relation it then follows from $(\initsecret,\sigma'.\kappa)\sim_a(\initsecretp,\tau'.\kappa')$ that $(\initsecret,\sigma')\sim_a(\initsecretp,\tau')$ so that by inductive assumption we already have $\initsecret[\sigma']_a = \initsecretp[\tau']_a$.

We now distinguish the different cases. If $\kappa = ab$ and $\kappa'$ is $ab$ or $a^eb$, then from $(\initsecret,\sigma'.\kappa)\sim_a(\initsecretp,\tau'.\kappa')$ we also obtain that $\initsecret[\sigma']_b = \initsecretp[\tau']_b$. From that and $\initsecret[\sigma']_a = \initsecretp[\tau']_a$ we then obtain that $\initsecret[\sigma']_a \union \initsecret[\sigma']_b = \initsecretp[\tau']_a \union \initsecretp[\tau']_b$ and therefore also, as $a$'s knowledge is the same in $(\initsecret,\sigma')$ and $(\initsecretp,\tau')$, that $\initsecret[\sigma']_a \union \initsecret[\sigma']_b{\setminus}* = \initsecretp[\tau']_a \union \initsecretp[\tau']_b{\setminus}*$. Furthermore, as $(\initsecret,\sigma') \sim_a (\initsecretp,\tau')$ and $\initsecret[\sigma']_b = \initsecretp[\tau']$, we also have that $(\initsecret[\sigma']_a \union \initsecret[\sigma']_b{\setminus}*){\setminus}{**} = (\initsecretp[\tau']_a \union \initsecretp[\tau']_b{\setminus}*){\setminus}{**}$, which means by definition that $\initsecret[\sigma'.\kappa]_a = \initsecretp[\tau'.\kappa]_a$. 

If $\kappa'$ is $ab^e$ we proceed almost as before, except that from $(\initsecret,\sigma'.\kappa)\sim_a(\initsecretp,\tau'.\kappa')$ we now obtain that $\initsecret[\sigma']_b = \initsecretp[\tau']_b^{\pm{e}}$, and then determine $a$'s novel holding of values in $(\initsecretp,\tau'.ab^e)$ as $(\initsecretp[\tau']_a \union \initsecretp[\tau']_b^{\pm{e}}{\setminus}*){\setminus}{**}$, wherein the value of $e$ is swapped.

All other cases for calls $\kappa$ involving $a$ and $b$ proceed similarly (using symmetric closure of the definition of the observation relation).

Let finally $\sigma=\sigma'.bc^\kappa$ where call $bc^\kappa$ involves agents $b$ and $c$ (in either order) different from $a$. Then $\tau$ must have shape $\tau=\tau'.de^{\kappa'}$ for some $d,e \neq a$. From $(\initsecret,\sigma'.bc^\kappa)\sim_a(\initsecretp,\tau'.de^{\kappa'})$ we then obtain that $(\initsecret,\sigma')\sim_a(\initsecretp,\tau')$ so that by inductive assumption we already have $\initsecret[\sigma']_a = \initsecretp[\tau']_a$ and therefore as well $\initsecret[\sigma'.bc^\kappa]_a = \initsecretp[\tau'.de^{\kappa'}]_a$, as required.
\end{proof}

As our structures are distributed systems where all propositional variables are local to agents, the truth value of an atom local to agent $a$ (`held by $a$') is known by that agent. This is also entirely as expected.

\begin{lemma}[Locality] \label{lemma.localll}
$\models b_a \imp K_a b_a$, $\models \neg b_a \imp K_a \neg b_a$, $\models \ov{b}_a \imp K_a \ov{b}_a$, and $\models \neg \ov{b}_a \imp K_a \neg \ov{b}_a$. 
\end{lemma}
\begin{proof}
According to the logical semantics variables $b_a$ and $\ov{b}_a$ are local to agent $a$, that is $\initsecret,\sigma\models b_a$ iff $b \in \initsecret[\sigma]_a$ and $\initsecret,\sigma\models \ov{b}_a$ iff $\ov{b} \in \initsecret[\sigma]_a$. The required then follows from that, the semantics of knowledge, and Lemma~\ref{lemma.simtolocal} that $(\initsecret,\sigma) \sim_a (\initsecretp,\tau)$ implies $\initsecret[\sigma]_a = \initsecretp[\tau]_a$.
\end{proof}

The next Lemma~\ref{lemma.stubborn} says that \emph{own secrets} (the value of secret $a$ held by agent $a$) are preserved. This is a consequence of the modelling assumption that all agents initially only hold their own secret. We call this \emph{stubbornness} as it implies that even when later confronted with another value for their own secret, agents will always refuse to accept that incorrect value. Note that we do not require or stipulate this property in our semantics, but that it can be shown, it is a consequence of the semantics as defined.
\begin{lemma}[Stubbornness] \label{lemma.stubborn}
(i) $\initsecret, \epsilon \models a_a$ implies $\initsecret, \sigma \models a_a$, and  
(ii) $\initsecret, \epsilon \models \ov{a}_a$ implies $\initsecret, \sigma \models \ov{a}_a$. 
\end{lemma}
\begin{proof}
We show the first, where the second follows similarly. The proof is by induction on the length of call sequences $\sigma$. Assume $\initsecret, \epsilon \models a_a$. For $\sigma=\epsilon$ it is trivial.
Let now $\sigma=\tau.\kappa$. By induction we may assume that $\initsecret, \epsilon \models a_a$ implies $\initsecret, \sigma \models a_a$. Therefore $\initsecret, \sigma \models a_a$. We now have to show that for an arbitrary next call $\kappa$, $\initsecret, \sigma.\kappa \models a_a$. Consider Definition~\ref{def.semanticscall} of the call semantics. If agent $a$ is not involved in $\kappa$, $\initsecret[\sigma.\kappa]_a = \initsecret[\sigma]_a$. Therefore, as $a \in \initsecret[\sigma]_a$, also $a \in \initsecret[\sigma.\kappa]_a$. If agent $a$ is involved in $\kappa$, we take the union of the holdings of agent $a$ and of the other agent involved in that call, let us say $b$, with the exception of the values withheld in the $*$ and the $**$ sets. So, if we can prove that secret value $a$ is never a member of $*$ or $**$, it must be in that union, so that again $a \in \initsecret[\sigma.\kappa]_a$. Concerning $*$, the secret value withheld is the value that is different from the value that is known by agent $a$. We now observe that from $\initsecret, \sigma \models a_a$ and Locality Lemma~\ref{lemma.localll} it follows that $\initsecret, \sigma \models K_a a_a$. Therefore $a \notin *$. Concerning $**$, the secret value withheld is the value that is different from the value such that for all $(\initsecretp,\tau)$ such that $(\initsecretp,\tau) \sim_a (\initsecret,\sigma)$ and another requirement on $b$, gossip state $(\initsecretp,\tau)$ makes true that value. As $\initsecret, \sigma \models K_a a_a$, therefore $\initsecretp, \tau \models a_a$. So $a \notin **$. As $a \notin *$ and $a \notin **$, again we conclude that $a \in \initsecret[\sigma.\kappa]_a$.
\end{proof}
As the result holds for value $a$ and for value $\ov{a}$, and as an agent initially only holds a single value for its own secret, therefore an agent always only holds a single value for its own secret. A different way to describe this result is to say that for all initial secret distributions $\initsecret$ and call sequences $\sigma$, $\initsecret_a =  \initsecret[\sigma]_a^{\inter a}$.

\begin{lemma}[Knowledge of secrets is preserved] \label{lemma.preservationofknowledge}
(i) $\initsecret,\sigma\models K_a b_b$ and $\sigma \sqsubseteq \tau$ implies $\initsecret,\tau\models K_a b_b$, and
(ii) $\initsecret,\sigma\models K_a \ov{b}_b$ and $\sigma \sqsubseteq \tau$ implies $\initsecret,\tau\models K_a \ov{b}_b$.
\end{lemma}
\begin{proof}
We show the first, where the proof of the second is similar. The proof is by induction on the length of $\tau{\setminus}\sigma$. If $\tau{\setminus}\sigma=\epsilon$ then $\sigma=\tau$ and it is trivial. Let us now show this for $\tau.\kappa$, on the inductive assumption it holds for $\tau$ with $\sigma\sqsubseteq\tau$. Therefore assume arbitrary $(\initsecret',\tau'')$ such that $(\initsecret',\tau'') \sim_a (\initsecret,\tau.\kappa)$. If agent $a$ is not involved in $\kappa$, from the definition of $\sim_a$ it follows that we already had $(\initsecret',\tau'') \sim_a (\initsecret,\tau)$. From that and inductive assumption $\initsecret,\tau\models K_a b_b$ we then get $\initsecret',\tau''\models b_b$, and therefore $\initsecret,\tau.\kappa \models K_a b_b$. If agent $a$ is involved in $\kappa$, then $\tau''$ must have shape $\tau'.\kappa'$. We then have $(\initsecret',\tau'.\kappa') \sim_a (\initsecret,\tau.\kappa)$, and from the definition of $\sim_a$ it then follows that we already must have had that $(\initsecret',\tau') \sim_a (\initsecret,\tau)$. From that and inductive assumption $\initsecret,\tau\models K_a b_b$ it now follows that for all such $(\initsecret',\tau')$ we have that $\initsecret',\tau' \models b_b$. From that and Stubbornness Lemma~\ref{lemma.stubborn} it then follows that $\initsecret',\tau'.\kappa' \models b_b$. As $(\initsecret',\tau'.\kappa')$ was arbitrary, therefore $\initsecret,\tau.\kappa \models K_a b_b$.
\end{proof}
As a corrollary of Lemma~\ref{lemma.preservationofknowledge} we obviously have that:
\begin{corollary} \label{corollary.preservationofknowledge}
$\initsecret,\sigma\models \Kv_a b$ and $\sigma \sqsubseteq \tau$ implies $\initsecret,\tau\models \Kv_a b$.
\end{corollary}
This seems to describe Lemma~\ref{lemma.preservationofknowledge} more succinctly and intuitively. But it is slightly weaker, as the formulation does not rule out that $\initsecret,\sigma\models K_a b_b$ whereas $\initsecret,\sigma.\kappa\models K_a \ov{b}_b$. 

Although knowledge of secrets is preserved, belief of secrets may not be preserved: if $K_a b_b$ is true now it remains true forever, whereas if $K_a b_a$ is true now it may be false after error correction such that $K_a \neg b_a$ is true later (and such that $K_a b_a$ is then false, as this is inconsistent with $K_a \neg b_a$).

\medskip

In the next Proposition~\ref{prop.kcorrectb} we show that knowledge implies correct belief. In the proof of item $(i)$ of this proposition we need a lemma, that we therefore present first. This Lemma~\ref{wieditleestisgek} requires additional terminology.

For arbitrary initial secret distributions $\initsecret$ we define $\initsecret^{\pm{b}}$ as the initial secret distribution that is like $\initsecret$ except that agent $b$ holds the other value of its own secret, that is, using our already defined notation for swapping values in secret holdings, $(\initsecret^{\pm{b}})_b = \initsecret_b^{\pm{b}}$ ($= \initsecret[\epsilon]_b^{\pm{b}}$), and for all $d \neq b$, $(\initsecret^{\pm{b}})_d = \initsecret_d$. Furthermore, for arbitrary call sequences $\sigma$ we define $\cor_b(\sigma)$ as the call sequence that is either $\sigma$ or wherein, if $\sigma$ contains a faulty call with a transmission error for (the value of) secret $b$, this faulty call is replaced by a correct call between the same agents: $\cor_b(\epsilon):=\epsilon$, $\cor_b(\sigma.ac) := \cor_b(\sigma).ac$, $\cor_b(\sigma.a^dc) := \cor_b(\sigma).a^dc$ and $\cor_b(\sigma.ac^d) := \cor_b(\sigma).ac^d$ (where $d \neq b$), $\cor_b(\sigma.a^bc) = \cor_b(\sigma.ac^b) := \cor_b(\sigma).ac$.\footnote{For example, $\cor_b(ab.a^bc.cd) = ab.ac.cd$, wherein the transmission error for secret $b$ is corrected; $\cor_b(ab.ac^c{\!\!}.cd) = ab.ac^c{\!\!}.cd$, although there was a transmission error, it was for another secret than $b$; $\cor_b(ab.ac.cd) = ab.ac.cd$, as there was no transmission error.} 

\begin{lemma} \label{wieditleestisgek}
For all initial secret distributions $\initsecret$, call sequences $\sigma$, and agents $a,b \in A$: 
If $\initsecret_b = \{b\}$, and $\initsecret[\sigma]_a^{\inter b}= \emptyset$ or $\initsecret[\sigma]_a^{\inter b}= \{\ov{b}\}$, then $(\initsecret, \sigma) \sim_a (\initsecret^{\pm{b}}, \cor_b(\sigma))$.\footnote{Obviously we also have, with the roles of $b$ and $\ov{b}$ swapped: If $\initsecret_b = \{\ov{b}\}$, and $\initsecret[\sigma]_a^{\inter b}= \emptyset$ or $\initsecret[\sigma]_a^{\inter b}= \{b\}$, then $(\initsecret, \sigma) \sim_a (\initsecret^{\pm{b}}, \cor_b(\sigma))$. By analogy that is needed in item $(ii)$ of Proposition~\ref{prop.kcorrectb}.}
\end{lemma}
\begin{proof}
The proof is by induction on the length of $\sigma$ and by further case distinction. The lemma's formulation and stubbornness implies that $a \neq b$.

{\bf Base case} $\bm{\epsilon}$ \\ We must have that $\initsecret[\sigma]_a^{\inter b}= \emptyset$. Also, $(\initsecret, \epsilon) \sim_a (\initsecret^{\pm{b}}, \epsilon)$ follows from the basic clause of $\sim_a$, because $\initsecret_a = \initsecret^{\pm{b}}_a = \{a\}$ or $\initsecret_a = \initsecret^{\pm{b}}_a = \{\ov{a}\}$: both singletons only contain the secret of $a$. Therefore, as $\cor_b(\epsilon) = \epsilon$, $(\initsecret, \epsilon) \sim_a (\initsecret^{\pm{b}}, \cor_b(\epsilon))$.

{\bf Inductive case} $\bm{\sigma.\kappa}$ \\ If $\kappa$ does not involve agent $a$, then $\initsecret[\sigma.\kappa]_a^{\inter b} = \initsecret[\sigma]_a^{\inter b}$, so that either way ($\emptyset$ or $\{\ov{b}\}$) by induction we obtain $(\initsecret, \sigma) \sim_a (\initsecret^{\pm{b}}, \cor_b(\sigma))$, and therefore as well, as $\kappa$ does not involve agent $a$, $(\initsecret, \sigma.\kappa) \sim_a (\initsecret^{\pm{b}}, \cor_b(\sigma).\cor_b(\kappa))$. Now it is easy to see that for arbitrary $\sigma$ and $\tau$, $\cor_b(\sigma.\tau) = \cor_b(\sigma).\cor_b(\tau)$. Therefore $(\initsecret, \sigma.\kappa) \sim_a (\initsecret^{\pm{b}}, \cor_b(\sigma.\kappa))$. 

Let now $\kappa$ involve agent $a$ and an arbitrary other agent $c$.

If $\initsecret[\sigma.\kappa]_a^{\inter b}= \emptyset$ then already $\initsecret[\sigma]_a^{\inter b}= \emptyset$ or $\initsecret[\sigma]_a^{\inter b}= \{\ov{b}\}$, and $\initsecret[\sigma]_c^{\inter b}= \emptyset$ or $\initsecret[\sigma]_c^{\inter b}= \{\ov{b}\}$ (we can rule out that either set contains $b$ as correct values are never discarded in the call semantics). By induction for agents $a$ and $c$ we then obtain that $(\initsecret, \sigma) \sim_a (\initsecret^{\pm{b}}, \cor_b(\sigma))$ respectively $(\initsecret, \sigma) \sim_c (\initsecret^{\pm{b}}, \cor_b(\sigma))$, and from the latter we obtain $\initsecret[\sigma]_c = \initsecret^{\pm{b}}[\cor_b(\sigma)]_c$. Now if $\kappa$ is a correct call, from $(\initsecret, \sigma) \sim_a (\initsecret^{\pm{b}}, \cor_b(\sigma))$ and $\initsecret[\sigma]_c = \initsecret^{\pm{b}}[\cor_b(\sigma)]_c$ we obtain $(\initsecret, \sigma.\kappa) \sim_a (\initsecret^{\pm{b}}, \cor_b(\sigma).\kappa)$ and therefore as in this case $\cor_b(\sigma).\kappa = \cor_b(\sigma.\kappa)$, $(\initsecret, \sigma.\kappa) \sim_a (\initsecret^{\pm{b}}, \cor_b(\sigma.\kappa))$. Whereas if $\kappa$ is a faulty call, it must be for another secret $d$ than $b$ (if $\kappa$ is faulty, $\sigma$ does not contain a faulty call, so we can rule out that $a$ or $c$ holds $\ov{b}$, and therefore neither $a$ nor $c$ has information about $b$), from $(\initsecret, \sigma) \sim_a (\initsecret^{\pm{b}}, \cor_b(\sigma))$ and $\initsecret[\sigma]_c^{\pm d} = \initsecret^{\pm{b}}[\cor_b(\sigma)]_c^{\pm d}$ we obtain $(\initsecret, \sigma.\kappa) \sim_a (\initsecret^{\pm{b}}, \cor_b(\sigma).\kappa)$, and as again (the error is in $b \neq d$) $\cor_b(\sigma).\kappa = \cor_b(\sigma.\kappa)$ we again obtain $(\initsecret, \sigma.\kappa) \sim_a (\initsecret^{\pm{b}}, \cor_b(\sigma.\kappa))$.

If $\initsecret[\sigma.\kappa]_a^{\inter b}= \{\ov{b}\}$ there are more subcases to consider. 

In case $\initsecret[\sigma]_a^{\inter b}= \emptyset$ and $\initsecret[\sigma]_c^{\inter b}= \{\ov{b}\}$, by induction for agents $a$ and $c$ we obtain that $(\initsecret, \sigma) \sim_a (\initsecret^{\pm{b}}, \cor_b(\sigma))$ respectively $(\initsecret, \sigma) \sim_c (\initsecret^{\pm{b}}, \cor_b(\sigma))$, from the latter we obtain $\initsecret[\sigma]_c = \initsecret^{\pm{b}}[\cor_b(\sigma)]_c$, so that $(\initsecret, \sigma.\kappa) \sim_a (\initsecret^{\pm{b}}, \cor_b(\sigma).\kappa)$ and therefore also (call $\kappa$ must have been correct as agent $c$ already held a faulty value of a secret before the call, so that again $\cor_b(\sigma).\kappa = \cor_b(\sigma.\kappa)$), $(\initsecret, \sigma.\kappa) \sim_a (\initsecret^{\pm{b}}, \cor_b(\sigma.\kappa))$. 

In the other case, $\initsecret[\sigma]_a^{\inter b}= \emptyset$, $\initsecret[\sigma]_c^{\inter b}= \{b\}$ and call $\kappa$ is faulty such that $\initsecret[\sigma]_c^{\pm{b}}$ (which contains $\ov{b}$) is the set of secrets passed on to agent $a$, and we proceed differently. First, from $\initsecret[\sigma]_a^{\inter b}= \emptyset$ and induction for $a$ it follows that $(\initsecret, \sigma) \sim_a (\initsecret^{\pm{b}}, \cor_b(\sigma))$. Second, as $\kappa$ is faulty, there are no faulty calls in $\sigma$, so that $\cor_b(\sigma) = \sigma$ and $\cor_b(\sigma.\kappa) = \sigma.\cor_b(\kappa)$. Furthermore, $\initsecret[\sigma]_c^{\pm b} = \initsecret^{\pm b}[\sigma]_c$. From $(\initsecret, \sigma) \sim_a (\initsecret^{\pm{b}}, \cor_b(\sigma))$ and $\cor_b(\sigma) = \sigma$ we get $(\initsecret, \sigma) \sim_a (\initsecret^{\pm{b}}, \sigma)$, from that and $\initsecret[\sigma]_c^{\pm b} = \initsecret^{\pm b}[\sigma]_c$ we get $(\initsecret,\sigma.\kappa) \sim_a (\initsecret^{\pm b}, \sigma.\cor_b(\kappa))$, and with the above $\cor_b(\sigma.\kappa) = \sigma.\cor_b(\kappa)$ therefore also $(\initsecret,\sigma.\kappa) \sim_a (\initsecret^{\pm b}, \cor_b(\sigma.\kappa))$ as requested.

In case $\initsecret[\sigma]_a^{\inter b}= \{\ov{b}\}$, call $\kappa$ must have been a correct call (as the faulty call, of which there is not more than one, must be in $\sigma$) so that either $\initsecret[\sigma]_c^{\inter b}= \emptyset$ or $\initsecret[\sigma]_c^{\inter b}= \{\ov{b}\}$. We now proceed as in the previous cases to again obtain $(\initsecret, \sigma.\kappa) \sim_a (\initsecret^{\pm{b}}, \cor_b(\sigma.\kappa))$.
\end{proof}

\begin{proposition}[Knowledge implies correct belief] \label{prop.kcorrectb} \label{cor.knowledgecorrectbelief}
(i) $\models K_a b_b \imp (b_b \et b_a \et \neg \ov{b}_a)$ and (ii) $\models K_a \ov{b}_b \imp (\ov{b}_b \et \neg b_a \et \ov{b}_a)$. 
\end{proposition}
\begin{proof}
We show the first, where the second is shown similarly. We show this by induction on the length of call sequences. 

For $\sigma=\epsilon$ it follows ex falso for $b \neq a$ and otherwise when $b=a$ from the  trivial observation that $b_b=b_a=a_a$, the instantiation $K_a a_a \imp a_a$ of the {\bf T} axiom, and the fact that agents initially only hold a unique value of their own secret, so that $\ov{a}_a$ is false.

For $\sigma=\tau.\kappa$ we distinguish the case where $\kappa$ does not involve agent $a$ from the case where $\kappa$ involves $a$.

For $\sigma=\tau.bc^\kappa$ where $b,c \neq a$, we assume that $\initsecret,\sigma.bc^\kappa\models K_a b_b$. We now have that $\initsecret,\sigma.bc^\kappa\models K_a b$, iff (semantics of formulas and observation relation) $\initsecret,\sigma\models K_a b$, which implies (induction) $\initsecret,\sigma\models b_b \et b_a \et \neg \ov{b}_a$, which implies (by stubbornness for $b_b$ and by the call semantics for non-involved agents) $\initsecret,\sigma.bc^\kappa \models b_b \et b_a \et \neg \ov{b}_a$.

For $\sigma=\tau.ac^\kappa$, assume $\initsecret,\sigma.ac^\kappa\models K_a b_b$. We distinguish two subcases: $\initsecret,\sigma\models K_a b_b$ and $\initsecret,\sigma\not\models K_a b_b$.

If $\initsecret,\sigma\models K_a b_b$, by induction we can conclude that $\initsecret,\sigma\models b_b \et b_a \et \neg \ov{b}_a$. Again from stubbornness and again from the call semantics, but now for involved agents (value $b$ for $a$ is preserved in the union, whereas value $\ov{b}$ either remains absent, or in case contributed by agent $c$ is discarded because in the $*$ set), we conclude that 
 $\initsecret,\sigma.ac^\kappa\models b_b \et b_a \et \neg \ov{b}_a$. So in this case $*$ removal may be involved.

If $\initsecret,\sigma\models \neg K_a b_b$, this is the harder case, and the case of most interest in the proof. 

 We show the three conjuncts separately, where the (hardest) second comes last.
\begin{itemize}

\item $\bm{\initsecret,\sigma.ac^\kappa\models b_b.}$ \quad From $\initsecret,\sigma.ac^\kappa\models K_a b_b$ and the validity of $K_a b_b \imp b_b$ (factivity {\bf T}) it follows that $\initsecret,\sigma.ac^\kappa\models b_b$.

\item $\bm{\initsecret,\sigma.ac^\kappa\models \neg\ov{b}_a.}$ \quad From $\initsecret,\sigma.ac^\kappa\models K_a b_b$ it follows that $\initsecret,\sigma.ac^\kappa\models \neg \ov{b}_a$ by the following line of reasoning. By definition, $\initsecret,\sigma.ac^\kappa\models K_a b_b$ means that (where $ac^{\kappa'}$ is `another call involving $a$ and $c$'): for all $(\initsecret, \sigma.ac^\kappa) \sim_a (\initsecretp, \tau.ac^{\kappa'})$, $\initsecretp,\tau.ac^{\kappa'}\models b_b$. By definition of the observation relation, $(\initsecret, \sigma.ac^\kappa) \sim_a (\initsecretp, \tau.ac^{\kappa'})$ means that $(\initsecret,\sigma)\sim_a (\initsecretp,\tau)$ and $\initsecret[\sigma]_c = \initsecretp[\tau]_c$, where, seemingly to complicate matters, one or both of these secret distributions may involve a swap $\initsecretp[\tau]_c^{\pm b}$. However, this is not a real complication, because we now satisfy the condition for $\ov{b}$ removal from agent $a$'s holding in the semantics of call. Which shows that, whether actually removed or not, $\initsecret,\sigma.ac^\kappa\models \neg \ov{b}_a$. So in this case $**$ removal may be involved.\footnote{We recall that $**$ implicitly only removes an incorrect value, in this case $\ov{b}$, if the correct value, in this case $b$, is known to be held by $b$, but that we cannot formalize this as truth of $K_a b_b$ as that would make the logical semantics circular. But in fact the knowledge is already there and `preserved' after the ** self-correction, and vice versa; the self-correction only concerns the truth of $b_a$ and $\ov{b}_a$, not of $b_b$.}

\item $\bm{\initsecret,\sigma.ac^\kappa\models b_a.}$ \quad It remains to show that $\initsecret,\sigma.ac^\kappa\models K_a b_b$ implies $\initsecret,\sigma.ac^\kappa\models b_a$. Now if $\initsecret,\sigma\models b_a$ or $\initsecret,\sigma\models b_c$ (and the call is correct), by the semantics of calls we immediately have that $\initsecret,\sigma.ac^\kappa\models b_a$, because value $b$ is then in the union $\initsecret[\sigma]_a \union \initsecret[\sigma]_c$ of the holdings of agents $a$ and $c$ (the $*$ deletion cannot occur because of assumption $\initsecret,\sigma\models \neg K_a b_b$ in this case of the proof; if a $**$ deletion occurs this involves value $\ov{b}$ because of proof assumption $\initsecret,\sigma.ac^\kappa\models K_a b_b$; see previous subcase).

We therefore only have the following remaining cases to consider: in $(\initsecret,\sigma)$, agent $a$ has no information on $b$ ($b_a$ and $\ov{b}_a$ are both false) or agent $a$ holds $\ov{b}$ ($\ov{b}_a$ is true), and agent $c$ has no information on $b$ or agent $c$ holds (or contributes, in a faulty call) $\ov{b}$. In all such cases either $\initsecret[\sigma.ac^\kappa]_a = \emptyset$ or  $\initsecret[\sigma.ac^\kappa]_a = \{\ov{b}\}$. Applying Lemma~\ref{wieditleestisgek} it then follows that $(\initsecret, \sigma.ac^\kappa) \sim_a (\initsecret^{\pm{b}}, \cor_b(\sigma.ac^{\kappa'}))$. As $(\initsecret^{\pm{b}}, \sigma.ac^{\kappa'}) \not \models b_b$, therefore $(\initsecret, \sigma.ac^\kappa) \not \models K_a b_b$. Therefore in all these cases it cannot be that $\initsecret,\sigma.ac^\kappa\models K_a b_b$, which ends the proof.
\end{itemize}
\end{proof}

Two observations: $(i)$ In all four cases where $a$ and $c$ have no information on $b$ or incorrect information on $b$ we have that $\initsecret,\sigma \not\models K_a b_b$ as well as $\initsecret,\sigma.ac^\kappa \not\models K_a b_b$, so $a$'s ignorance of secret $b$ then persists. $(ii)$ We can only have that $\initsecret,\sigma \models \neg K_a b_b \et \neg b_a \et \neg \ov{b}_a \et b_c$ and  
$\initsecret,\sigma.ac^\kappa \models K_a b_b$ when agent $a$ already self-corrected on a secret $d \neq b$ in the past (during $\sigma$), and therefore knows that any value she receives for a secret about which she is uninformed must be the correct value.  This is the case of the proof above where it says ``if $\initsecret,\sigma\models b_a$ or $\initsecret,\sigma\models b_c$'' (namely when $\initsecret,\sigma\not\models b_a$).

A succinct way to express Proposition~\ref{prop.kcorrectb} is that
$\initsecret,\sigma\models \Kv_a b$ implies $\initsecret[\sigma]_a^{\inter b} = \initsecret_b$. 

As a corollary of Proposition~\ref{prop.kcorrectb} we have that {\em knowledge is justified correct belief}: 
\begin{corollary}[Knowledge is justified correct belief] \label{cor.kcorrectb2}\label{cor.knowledgecorrectbelief2}
$\models K_a b_b \eq K_a (b_b \et b_a \et \neg \ov{b}_a)$ and $\models K_a \ov{b}_b \eq K_a (\ov{b}_b \et \neg b_a \et \ov{b}_a)$. 
\end{corollary}
\begin{proof}
Although a corollary indeed, let us give the simple proof, where we consider the first. The direction from right to left is obvious. For the direction from left to right:
Given $(\initsecret,\sigma)$, assume $\initsecret,\sigma \models K_a b_b$. Let $(\initsecret,\sigma)\sim_a (\initsecretp,\tau)$. By the properties of knowledge we then not only have that $\initsecretp,\tau \models b_b$ but also that $\initsecretp,\tau \models K_a b_b$. For all those $\tau$ we can apply Proposition~\ref{prop.kcorrectb} and obtain $\initsecretp,\tau \models b_b \et b_a \et \neg \ov{b}_a$. As $(\initsecretp,\tau)$ was arbitrary we therefore have $\initsecret,\sigma \models K_a (b_b \et b_a \et \neg \ov{b}_a)$. Now using assumption $\initsecret,\sigma \models K_a b_b$ and that $(\initsecret,\sigma)$ was also arbitrary we get the required $\models K_a b_b \imp K_a (b_b \et b_a \et \neg \ov{b}_a)$. \end{proof}

\subsection{Examples}
We explain the semantics with some elementary examples.

\begin{example}  \label{example.dadada}
Consider two agents $a,b$ and the initial secret distribution $\initsecreto = a|b$.
\begin{itemize}

\item 
First consider a single call $ab$, so that $(a|b)[ab] = ab|ab$. Then agent $a$ holds the correct value of secret $b$: $a|b, ab \models b_b \et b_a$. However $a$ considers it possible that she holds the incorrect value of $b$: $a|b, ab \models \M_a (\ov{b}_b \et b_a)$, because $(a|b, ab) \sim_a (a|\ov{b}, ab^b)$ and $a|\ov{b}, ab^b \models \ov{b}_b \et b_a$. 

This item illustrates that true belief does not imply justified true belief, that is, knowledge: $a|b, ab \models b_b \et b_a \et \neg \ov{b}_b$ whereas $a|b, ab \not\models \Kv_a b$ ($a|b, ab \not\models K_a (b_b \et b_a \et \neg \ov{b}_b$).

\item
Then consider call sequence $ab.ab^b$. We now have that $(a|b,ab.ab^b) \sim_a (a|\ov{b},ab^b{\!\!}.ab)$, and no other gossip states are considered possible. Furthermore $(a|b)[ab.ab^b] = a\un{b}|ab$ and $(a|\ov{b})[ab^b{\!\!}.ab] = a\un{b}|a\ov{b}$. Either way, $(a|b)[ab.ab^b]_a = (a|\ov{b})[ab^b{\!\!}.ab]_a = a\un{b}$: agent $a$ has a conflict for the secret of agent $b$. In order to resolve the conflict, she needs to call $b$ again.

\item Therefore, now consider call sequence $ab.ab^b{\!\!}.ab$. Intuitively, agent $a$ now has three independent sources of information on the secret of $b$, of which the majority has value $b$. This is sufficient for her to learn the correct value $b$ of $b$'s secret. Because formally, the last observation further restricts the set of possible call sequences. From the two call sequences that $a$ considered possible before, only the first can be extended such that the first and third received value for $b$ correspond, we cannot extend the second into $ab^b{\!\!}.ab.ab^b$ as this contains more than one faulty call. The only gossip state she therefore considers possible is $(a|b,ab.ab^b{\!\!}.ab)$. At the third call agent $a$ can therefore restrict the set of initial secret distributions that she considers possible from $\{a|b,a|\ov{b}\}$ to the singleton $\{a|b\}$, such that she now knows that the correct value of $b$'s secret is $b$ and not $\ov{b}$. (Agent $a$  knows / has justified true belief of secret $b$.) 

The incorrect value $\ov{b}$ held by $a$ in $(a|b)[ab.ab^b] = (a\un{b}|ab)$ no longer appears in her holding in $(a|b)[ab.ab^b{\!\!}.ab] = (ab,ab)$ as a consequence of the deletion of that value according to the semantics of calls. In this case we have that, for agent $a$, ${**} = \{\ov{b}\}$ (and $* = \emptyset$).

Furthermore, also $b$ will not incorporate the incorrect value $\ov{b}$ of his own secret that he receives from $a$ in the third call. This is because knowledge $\Kv_b b$ of his own secret is preserved after all calls and therefore in particular after the the first two calls. In this case we have that, for agent $b$, ${*} = \{\ov{b}\}$ (and ${**} = \emptyset$).

\item As a different extension of the first call $ab$, consider $ab.ab$. Even though $(a|b,ab) \sim_a (a|\ov{b},ab^b)$, we have that $(a|b,ab.ab) \not\sim_a (a|\ov{b},ab^b{\!\!}.ab)$, because agent $a$ has a conflict in one and not in the other. Also, agent $a$ does not consider $(a|\ov{b},ab^b{\!\!}.ab^b)$, because $ab^b{\!\!}.ab^b$ has more than one faulty call. We have that $a|b, ab.ab\models K_a b_b$: after a majority of two independent observations of value $b$, $a$ again knows that this must be the correct value. We can also get knowledge of secrets without first having a conflict to resolve.

\item Finally, consider call sequence $ab.ab.ab^b$ extending the previous $ab.ab$. Receiving an incorrect value $\ov{b}$ for the secret of $b$ does not cause agent $a$ to have a conflict for the secret of $b$ as she already knows the (correct) value $b$ for that secret. As the value $\ov{b}$ is known to be incorrect it is not added to the holding of secret values of agent $a$ (the set ${*} = \{\ov{b}\}$), so that $(a|b)[ab.ab] = (a|b)[ab.ab.ab^b] = ab|ab$. (This is just as for agent $b$ after $ab.ab^b{\!\!}.ab$ in the first item.) 
\end{itemize}
\end{example}

\begin{example}
Let there now be three agents $a,b,c$.
\begin{itemize}
\item Consider initial gossip state $a|b|c$ and call sequence $ac.ab$. Agent $a$ considers the following gossip states possible: $(a|b|c, ac.ab)$, $(a|b|\ov{c}, ac^c{\!\!}.ab)$, $(a|\ov{b}|c, ac.ab^b)$, as well as well as other gossip states wherein a transmission error is made by herself. They all result in agent $a$ holding $\{a,b,c\}$. Agent $a$ does not know she holds the correct values of all three secrets. From $(a|b|c, ac.ab) \sim_a (a|b|\ov{c}, ac^c{\!\!}.ab)$ we obtain that $ac.ab \models \M_a \ov{c}_c$. 
\item Consider the extension $ac.ab.ab^b$. We have that $(a|b|c)[ac.ab.ab^b]= a\un{b}c|abc|ac$. Agent $a$ now still considers possible the gossip states $(a|b|c, ac.ab.ab^b)$ and $(a|\ov{b}|c,ac.ab^b{\!\!}.ab)$. She no longer considers possible that she holds an incorrect value for $c$. She knows that at most one transmission error occurs, and knows that the error involved the secret of $b$, she therefore now knows the secret of $c$: $a|b|c, ac.ab.ab^b \models K_a c_c$, implying $a|b|c, ac.ab.ab^b \models \Kv_a c$. 
\end{itemize}
\end{example}

\begin{example}
An agent can get to know the secret of another agent without ever calling that agent. Consider the call sequence $ab.bc.ad.de.ce$ and the usual initial secret distribution $\initsecreto$ (that is, $a|b|c|d|e$). We show that after this call sequence agent $e$ knows the secret of $a$, without ever having called $a$. Let us first show schematically how the secret distribution develops and after that justify some details:

\medskip

\noindent 
$
a|b|c|d|e \ \stackrel{ab.bc.ad}\imp \ abd|abc|abc|abd|e \ \stackrel{de}\imp \ abd|abc|abc|abde|abde \ \stackrel{ce}\imp \ abd|abc|abcde|abde|abcde$

\medskip

\noindent After call $de$ agent $e$ holds all secrets except $c$. At this stage agent $e$ considers it possible that she incorrectly believes $a$: for example, we have that $(a|b|c|d|e,ab.bc.ad.de) \sim_e (\ov{a}|b|c|d|e,ab.bc.ad.d^ae)$ (and that $(a|b|c|d|e,ab.bc.ad.de) \sim_e (\ov{a}|b|c|d|e,ab.bc.a^ad.de)$; note that an error in the first call $a^ab$ will not reach $d$ and therefore not $e$). After call $ce$ agent $e$ learns that $c$ did not hold a value for secret $d$ but was informed about $a$. As the values of $a$ received from $d$ and from $c$ do not conflict, agent $e$ can therefore rule out gossip state $(\ov{a}|b|c|d|e,ab.bc.ad.d^ae)$. For the same reason that the observed value of $a$ did not conflict in the final two calls $de$ and $ce$, she can also rule out that the final call $ce$ was an incorrect call $c^ae$, or that the second call $bc$ was an incorrect call $b^ac$. Finally, agent $c$ can rule out that the first call was $a^ab$, as that would have made $b$ pass on this incorrect value $\ov{a}$ to $c$ in call $bc$, whereas in call $ad$ agent $d$ would still have received value $a$ from $a$, so that also in that case $e$ would have observed different values for $a$ in calls $de$ and $ce$. Agent $e$ now knows the secret of $a$.

A fortiori, an agent can get to know the secret of another agent including self-correcting the value for the secret of that agent without ever calling that agent. A simple variation on the previous would be the call sequence $ab.bc.ad.d^ae.de.ce$. After $d^ae$, agent $e$ holds an incorrect value $\ov{a}$ of the secret of agent $a$. After calling $d$ again she receives correct value $a$ and thus now has a conflict for the secret of $a$. In call $ce$ agent $e$ learns, as before, that agent $c$ does not hold a value for secret $d$. The value of $a$ agent $e$ receives in that call is therefore an independent observation, which again creates a majority of $a$ over $\ov{a}$. 

Yet another example, for four agents $a,b,c,d$, is $(a|b|c|d,ab.a^ac.ad.cd.cb)$. We get the following transitions:

\medskip

\noindent 
$
a|b|c|d \ \stackrel{ab.a^ac}\imp \ abc|ab|\ov{a}bc|d \ \stackrel{ad}\imp \ abcd|ab|\ov{a}bc|abcd  \stackrel{cd}\imp  abcd|ab|\un{a}bcd|\un{a}bcd  \stackrel{cb}\imp  abcd|\un{a}bcd|abcd|\un{a}bcd$

\medskip

After $ab.a^ac$ agent $c$ knows that the first call was involving $a$ and $b$, however she cannot distinguish $(a|b|c|d,ab.a^ac)$ from $(\ov{a}|b|c|d,ab.ac)$. After the subsequent third and fourth calls $ad.cd$, agent $c$ (after the second call involving her) still cannot distinguish $(a|b|c|d,ab.a^ac.ad.cd)$ from $(\ov{a}|b|c|d,ab.ac.ad.cd^a)$. However, the final call $cb$ rules out secret distribution $\ov{a}|b|c|d$, as that would otherwise have resulted in agent $c$ learning that $a$ holds $\ov{a}$ instead of $a$. So, the alternative evolution of secret distributions would then have been:

\medskip

\noindent 
$
\dots \ \stackrel{ab.ac.ad}\imp \ \ov{a}bcd|\ov{a}b|\ov{a}bc|\ov{a}bcd \ \stackrel{cd^a}\imp \ \ov{a}bcd|\ov{a}b|\un{a}bcd|\ov{a}bcd \ \stackrel{cb}\imp \ \ov{a}bcd|\un{a}bcd|\ov{a}bcd|\ov{a}bcd$

\medskip

\noindent Clearly, agent $c$ again knows the secret of $a$. But of course she can distinguish a call sequence where she knows that its value is $a$ from a call sequence where she knows that its value is $\ov{a}$.
\end{example}

\subsection{Gossip protocol}

\paragraph*{Expert and super expert}
An agent who holds all secrets without conflict is an {\em expert}, an agent who knows that all agents hold all secrets without conflict is a {\em super expert}. Similarly, an agent who holds all correct secrets without conflict is a {\em correct expert} and an agent who knows that all agents hold all correct secrets without conflict is a {\em correct super expert}. Call sequences satisfying that all are (correct) experts or (correct) super experts are called \emph{(correct) successful} respectively \emph{(correct) supersuccessful}. In gossip without errors we only consider correct experts and correct super experts, and all success is correct success.

\paragraph*{Gossip protocol} 
A {\em gossip protocol} is an algorithm encoding the intuition: \begin{quote} {\em Until the {\bf termination condition} holds, select a pair $ab$ of agents that satisfy a {\bf call condition}, and execute call $ab$}. \end{quote} The \emph{gossip problem} is whether a gossip protocol terminates and for which condition. 

The termination condition is also known as the {\em epistemic goal}. The usual termination condition is that all agents are experts. We also consider the termination condition that all agents are super experts. Apart from that we consider the termination conditions that all are correct experts or all are correct super experts. 

The call condition for a call from $a$ to $b$ (whether correct or faulty) in general requires $a$ and $b$ to be \emph{neighbours} (given a possibly restricted network), and for agent $a$ to know that the call condition holds. However, here we only consider that all agents are neighbours and we only consider the gossip protocol $\ANY$ with trivial call condition $\top$ (`true'), which is known by agent $a$ --- another triviality, as $K_a \top$ is valid.   

There are also more distributed, equivalent, ways to describe gossip protocols than above \cite{AptW18,DitmarschEPRS17,hvdetal.lucky:2024}. Instead of \emph{termination} it is sufficient to require \emph{stabilization}, such that any permitted call can continue to be executed forever, even when all agents are (super) experts. The effect of this is removing the `until' parts in the description above. Given an infinite call sequence, in order to have success of some kind, stabilization (termination) is only required if the infinite sequence is \emph{fair}: at any stage all permitted calls will occur again later \cite{DitmarschGR23,AptW18}. Stabilization is more in line with assumptions in distributed computing.

Our setting for gossip with errors seems a very simple one, because calls and protocols are not in the logical language, the call condition is the trivial formula $\top$, and the network is complete (all are neighbours). However, the presence of errors comes with novel logical complications, as self-correcting makes true propositional variables $b_a$ or $\ov{b}_a$ false again, unlike in error-free gossip. 

Let us now formally introduce the terminology to accommodate the novel distinction between correct and incorrect termination (note the disjunctions turn out exclusive).
%
%
%
\[ \begin{array}{llll}
\Exp_a & := & \Et_{b \in A} ((b_a \et \neg \ov{b}_a) \vel (\neg b_a \et \ov{b}_a)) \ \ & \text{$a$ is an expert} \\
\Exp_A & := & \Et_{a \in A} \Exp_a & \text{all agents are experts} \\
\cExp_a & := & \Et_{b \in A} ((b_b \et b_a \et \neg \ov{b}_a) \vel (\ov{b}_b \et \neg b_a \et \ov{b}_a)) \ & \text{$a$ is a correct expert} \\
\cExp_A & := & \Et_{a \in A} \cExp_a & \text{all agents are correct experts}
\end{array}\]
Furthermore, $K_a\Exp_A$ means that {\em agent $a$ is a super expert} and $K_a\cExp_A$ means that {\em agent $a$ is a correct super expert}. Given all that, 
termination condition $E_A\Exp_A$ requires that \emph{all are super experts} and $E_A\cExp_A$ that \emph{all are correct super experts}.
The gossip protocol is now called (where this should hold for arbitrary initial secret distributions $\initsecret$, $\sigma^\omega$ informally represents an infinite call sequence, and $\tau$ is a call sequence, that is, finite):
\begin{itemize}
\item \emph{successful} if all fair $\sigma^\omega$ have a prefix $\tau$ such that $\initsecret,\tau \models \Exp_A$;
\item \emph{supersuccessful} if all fair $\sigma^\omega$ have a prefix $\tau$ such that $\initsecret,\tau \models E_A\Exp_A$.
\item \emph{correct successful} if all fair $\sigma^\omega$ have a prefix $\tau$ such that $\initsecret,\tau \models \cExp_A$;
\item \emph{correct supersuccessful} if all fair $\sigma^\omega$ have a prefix $\tau$ such that $\initsecret,\tau \models E_A\cExp_A$;
\item \emph{first-correct successful} if $\models \Exp_A \imp \cExp_A$;
\item \emph{first-correct supersuccessful} if $\models E_A\Exp_A \imp E_A\cExp_A$.
\end{itemize}
The last four are novel categories of termination. Concerning the final two, it not only matters whether all become correct experts but also whether this happens first. Note that $\Exp_A$ may be true, but then become false again when an agent after a subsequent call has a conflict for the value of some secret; and subsequently may again become true, as well as, eventually, $\cExp_A$. In view of that it may be considered remarkable that $E_A\cExp_A$ is the only stable termination goal (Proposition~\ref{prop.stable}), as it is equivalent to $\Et_{a,b} \Kv_a b$ in our semantics. Preservation of truth is then a direct consequence of Lemma~\ref{lemma.preservationofknowledge}. We continue with such results for successful and supersuccessful termination. 

\subsection{Results for successful and supersuccessful termination} 
We will show that eventually everybody becomes an expert and a correct expert, but that it cannot be guaranteed that everybody becomes a correct expert before becoming an expert.

\begin{proposition} \label{expandcexp} $(1.)$ $\models \Exp_a \eq K_a \Exp_a$  and $(2.)$ $\models \neg\Exp_a \eq K_a \neg\Exp_a$ but $(3.)$ $\not\models \cExp_a \eq K_a \cExp_a$.
\end{proposition}

\begin{proof} The direction $K_a\phi \imp \phi$ of all the above is a property of knowledge. For the other direction:
\begin{enumerate}
\item This follows almost directly from Lemma~\ref{lemma.localll} that the value of all local atoms is known by the agent holding them. If $\Exp_a$, then for all $b \in A$, $(b_a \et \neg \ov{b}_a) \vel (\neg b_a \et \ov{b}_a)$ is true. From that, with Lemma~\ref{lemma.localll}, also follows $(K_a b_a \et K_a\neg \ov{b}_a) \vel (K_a\neg b_a \et K_a\ov{b}_a)$. With the properties of knowledge we then also have that $K_a(b_a \et \neg \ov{b}_a) \vel K_a(\neg b_a \et \ov{b}_a)$, and by weakening both known formulas we then obtain $K_a((b_a \et \neg \ov{b}_a) \vel (\neg b_a \et \ov{b}_a))$ (in disjunction with itself, so we omitted that).
\item If $\neg \Exp_a$, there must be a secret that $a$ does not hold or for which she has a conflict, that is, $\neg b_a \et \neg \ov{b}_a$ or $b_a \et \ov{b}_a$. Using Lemma~\ref{lemma.localll} of locality again, we then have $K_a\neg b_a \et K_a\neg \ov{b}_a$ or $K_a b_a \et K_a \ov{b}_a$, and therefore $K_a(\neg b_a \et \neg \ov{b}_a)$ or $K_a (b_a \et \ov{b}_a)$. Therefore $K_a \neg \Exp_a$.
\item However, $\not\models \cExp_a \eq K_a \cExp_a$. Given four agents $a,b,c,d$, a very simple counterexample is that $a|b|c|d, ab.cd.ac.bd \models \cExp_a$ whereas $a|b|c|d, ab.cd.ac.bd \not\models K_a \cExp_a$ because $(a|b|c|d, ab.cd.ac.bd) \sim_a (a|b|\ov{c}|d, ab.cd.ac^c{\!\!}.bd)$, and $(a|b|\ov{c}|d)[ab.cd.ac^c{\!\!}.bd]_a$ is also $abcd$ ($\{a,b,c,d\}$) but in that case that is incorrect, so that $a|b|\ov{c}|d, ab.cd.ac^c{\!\!}.bd \not\models \cExp_a$. It is easy to see that also $b,c,d$ do not know that knowledge of all secrets is correct after $ab.cd.ac.bd$. A more involved counterexample is Example~\ref{ONE} in the next subsection.
\end{enumerate}
\end{proof}

\begin{proposition} \label{anyresults}
The gossip protocol $\ANY$ is: (1.) successful, (2.) supersuccessful, (3.) correct successful, and (4.) correct supersuccessful. However, it is (5.) not first-correct successful and (6.) not first-correct supersuccessful.
\end{proposition}

\begin{proof} We recall that the notions of successful and supersuccessful are defined with respect to fairly scheduled infinite call sequences, and that in the protocol $\ANY$ any call $ab$ remains arbitrarily often permitted (instead of termination we assume stabilization).

Without loss of generality, assume initial secret distribution $\initsecreto$. For items (1.) to (4.) we show that for any gossip state $(\initsecreto,\sigma)$ not satisfying the termination condition, call sequence $\sigma$ can be extended to a $\tau$ such that $(\initsecreto,\tau)$ satisfies the termination condition.
\begin{enumerate}
\item $\ANY$ is successful: Given $\initsecreto,\sigma\not\models\Exp_A$, there must be $a \neq b$ with $\initsecreto[\sigma]_a^{\inter b} = \emptyset$ or $\initsecreto[\sigma]_a^{\inter b} = \{b,\ov{b}\}$. In the first case, after $\sigma$, call $ab$ is permitted so that $\initsecreto[\sigma.ab]_a^{\inter b} = \{b\}$  and therefore $\initsecreto,\sigma.ab\models (b_a \et \neg \ov{b}_a) \vel (\neg b_a \et \ov{b}_a)$, as required. In the second case call $ab$ is also permitted and results in error correction so that also $\initsecreto[\sigma.ab]_a^{\inter b} = \{b\}$. In case there are still such $a,b$ with $\initsecreto[\sigma]_a^{\inter b} = \emptyset$ or $\initsecreto[\sigma]_a^{\inter b} = \{b,\ov{b}\}$ we repeat the procedure.

\item $\ANY$ is supersuccessful: Given $\initsecreto,\sigma\not\models E_A\Exp_A$, there must be $a,b \in A$ with $\initsecreto,\sigma\models\M_a \neg \Exp_b$. Therefore, there is $(\initsecretp,\tau)\sim_a(\initsecreto,\sigma)$ with $\initsecretp,\tau \not\models \Exp_b$, and a $c \in A$ (possibly $c=a$) with $\initsecretp,\tau\models \neg c_b \et \neg \ov{c}_b$ or $\initsecretp,\tau\models c_b \et \ov{c}_b$. After $\tau$, call $bc$ is permitted, so that, as for the first item, $\initsecretp,\tau.bc\models (c_b \et \neg \ov{c}_b) \vel (\neg c_b \et \ov{c}_b)$. Call $bc$ is also permitted after $\sigma$, as well as a subsequent call $ab$ (note that these calls $bc$ and $ab$ may cause a conflict in agent $a$ for some secret $d \neq c$, which would requires a further extension of the call sequence with a single call $ad$ when iterating the procedure). We now have that $\initsecreto,\sigma.bc.ab\models K_a (c_b \et \neg \ov{c}_b)$ or that $\initsecreto,\sigma.bc.ab\models K_a (\neg c_b \et \ov{c}_b)$, and therefore $\initsecreto,\sigma.bc.ab\models K_a ((c_b \et \neg \ov{c}_b) \vel (\neg c_b \et \ov{c}_b))$ as required. We repeat the procedure until $a$ knows that $b$ is an expert, and until there are no $a,b \in A$ with $\initsecreto,\sigma\models\M_a \neg \Exp_b$.

\item $\ANY$ is correct successful: without loss of generality, suppose there are $a,b$ such that $b_b \et b_a \et \neg\ov{b}_a$ is false. Let us assume that $b_a \et \neg\ov{b}_a$ are true, as we have already shown $\ANY$ to be successful, so that $b_b$ must be false. If so, as call $ab$ is permitted, after this call agent $a$ knows secret $b$ so that with Proposition~\ref{cor.knowledgecorrectbelief} we obtain $b_b \et b_a \et \neg\ov{b}_a$. We repeat the procedure until this holds for all pairs $c,d$ of agents in $A$. Note that for $c \neq b$ we may well have that $c_c \et c_a \et \neg\ov{c}_a$ is true but $\Kv_a c$ is false, which does not require any further calls from agent $a$ (unlike in the next item, where knowledge of secrets is required). 

\item $\ANY$ is correct supersuccessful: We may assume the previous item, so we can obtain $b_b \et b_a \et \neg\ov{b}_a$ for all agents $a$ and $b$. In case agent $a$ does not know this, as call $ab$ is permitted, it is  sufficient to have $a$ call $b$ in order to confirm that the value $b$ she holds is correct. Also note that this termination goal is stable: see Proposition~\ref{prop.stable}, below.
\item $\ANY$ is not first-correct successful: see Example~\ref{ONE}, below.
\item $\ANY$ is not first-correct supersuccessful: see Example~\ref{FIVE}, below.
\end{enumerate}
\end{proof}

Correct super success is the only stable property of a gossip protocol in this setting with errors. 
\begin{proposition} \label{prop.stable}
If $\initsecret,\sigma\models E_A \cExp_A$ and $\sigma \sqsubseteq \tau$, then $\initsecret,\tau\models E_A \cExp_A$
\end{proposition}
\begin{proof}
We observe that $E_A \cExp_A$ is equivalent to the conjunction of $K_a ((b_b \et b_a \et \neg \ov{b}_a) \vel (\ov{b}_b \et \neg b_a \et \ov{b}_a))$ for all agents $a,b \in A$. Such a conjunct is equivalent to $K_a (b_b \et b_a \et \neg \ov{b}_a) \vel K_a (\ov{b}_b \et \neg b_a \et \ov{b}_a)$ (because agents know their local propositions, Lemma~\ref{lemma.localll}), and applying Corollary~\ref{cor.knowledgecorrectbelief2}  therefore equivalent to $K_a b_a \vel K_a \ov{b}_a$ in our semantics, in other words, equivalent to $\Kv_a b$. We obtained that $E_A \Exp_A$ is equivalent to $\Et_{a,b \in A} \Kv_a b$. Now applying Lemma~\ref{lemma.preservationofknowledge}, we immediately obtain that $\initsecret,\sigma\models E_A \Exp_A$ and $\sigma \sqsubseteq \tau$ imply $\initsecret,\tau\models E_A \Exp_A$.
\end{proof}

\subsection{Examples of successful and supersuccessful termination} 

We give examples of successful and supersuccessful termination, including the phenomenon of {\em lucky calls} where an agent may learn that another agent is an expert without calling them, that comes with novel variations in this setting with errors.

\begin{example}\label{ONE}
First, success. Assume $\initsecreto = a|b|c|d$. Consider call sequence $ab^b{\!\!}.bc.bd.cd.ab.ab$. In the first call, $a$ receives a faulty value $\ov{b}$ of $b$'s secret. Then, $b,c,d$ exchange all secrets between them. Then $a$ calls $b$ again and now has a conflict for $b$, and finally makes another call wherein she self-corrects and now knows the secret of $b$. A certain road towards resolving a conflict about $b$ is calling $b$, even when there might be other roads achieving the same goal. Also note that agent $b$ does not get a conflict about his own secret in that final call $ab$: the value $\ov{b}$ he obtains from agent $a$ he knows to be incorrect and therefore discards. Schematically the execution is:

\medskip

\noindent $
a|b|c|d \quad \stackrel{ab^b}\imp \quad a\ov{b}|ab|c|d \quad \stackrel{bc.bd.cd}\imp \quad a\ov{b}|abcd|abcd|abcd \quad \stackrel{ab}\imp \quad \\ a\un{b}cd|abcd|abcd|abcd \quad \stackrel{ab}\imp \quad abcd|abcd|abcd|abcd $

\medskip

\noindent We have that $a|b|c|d, ab^b{\!\!}.bc.bd.cd.ab.ab \models \cExp_A$.

Now consider $ab^b{\!\!}.ac.cd.da.ab.ab$ wherein agent $a$ disseminates an incorrect value for $b$ to the other agents. Then, as before, $a$ learns in a call with $b$ that the value of $b$ may be faulty and confirms the correct value $b$ in another call $ab$ (again obtaining the required majority of two independently obtained values $b$ over one $\ov{b}$). But agents $c$ and $d$ still incorrectly believe that $\ov{b}$ is the secret of $b$. Schematically:

\medskip

\noindent $
a|b|c|d \quad \stackrel{ab^b}\imp \quad a\ov{b}|ab|c|d \quad \stackrel{ac}\imp \quad a\ov{b}c|ab|a\ov{b}c|d \quad \stackrel{cd.da}\imp \quad a\ov{b}cd|ab|a\ov{b}cd|a\ov{b}cd \quad \stackrel{ab}\imp \quad \\ a\un{b}cd|abcd|a\ov{b}cd|a\ov{b}cd  \quad \stackrel{ab}\imp \quad abcd|abcd|a\ov{b}cd|a\ov{b}cd $

\medskip

\noindent We have that $a|b|c|d, ab^b{\!\!}.ac.cd.da.ab.ab \models \Exp_A \et \neg \cExp_A$. 

Finally, consider $ab^b{\!\!}.bc.ad.cd.cb.db.ad.ab$. The first conflict now appears in call $cd$. Both $c$ and $d$ need to verify $b$'s secret. Calls $cb$ and $db$ may be in either order. In subsequent call $ad$, agent $a$ learns a conflicting value for $b$, but not agent $d$, who already knows the correct value of $b$ after the prior call $db$. Agent $a$ now calls $b$ to learn the correct value. We end up with $a|b|c|d, ab^b{\!\!}.bc.ad.cd.cb.db.ad.ab \models \cExp_A$.

\medskip

\noindent $
a|b|c|d \quad \stackrel{ab^b}\imp \quad a\ov{b}|ab|c|d \quad \stackrel{bc}\imp \quad a\ov{b}|abc|abc|d \quad \stackrel{ad}\imp \quad a\ov{b}d|abc|abc|a\ov{b}d \quad \stackrel{cd}\imp \quad \\ a\ov{b}d|abc|a\un{b}cd|a\un{b}cd  \quad \stackrel{cb.db}\imp \quad a\ov{b}d|abcd|abcd|abcd  \quad \stackrel{ad}\imp \quad \\ a\un{b}cd|abcd|abcd|abcd \quad \stackrel{ab}\imp \quad abcd|abcd|abcd|abcd $
\end{example}

\begin{example}\label{FIVE} 
Now, super success. First, consider the error-free eight call sequence \[ \sigma = ab.cd.ac.bd.ab.ad.bc.cd \] and the initial secret distribution $\initsecreto=a|b|c|d$. In error-free gossip this call sequence is supersuccessful and eight calls are optimal \cite{DitmarschGR23,hvdetal.lucky:2024}. In our semantics for gossip with errors the call sequence is successful and also correct successful as all agents hold the correct value: \[ \initsecreto,\sigma \models \Exp_A \et \cExp_A \] But it is not supersuccessful and therefore also not correct supersuccessful: \[ \initsecreto,\sigma \models \neg E_A\Exp_A \et \neg E_A\cExp_A \] 
For example, consider agent $a$. Agent $a$ considers it possible that the last call was $cd^d$, after which $c$ has a conflict for secret $d$ and is no longer an expert. From $(\initsecreto,\sigma) \sim_a (\initsecreto,ab.cd.ac.bd.ab.ad.bc.cd^d)$ it follows that $\initsecreto,\sigma \models \M_a \neg \Exp_c$ and therefore $\initsecreto,\sigma \not\models E_A \Exp_A$, and a fortiori also $\initsecreto,\sigma \not\models E_A \cExp_A$. 

In our setting the expert goal $\Exp_A$ has become `unstable', for example, extending $\sigma$ with the call $cd^d$ above also makes $c$ have a conflict, so that $\initsecreto,\sigma.cd^d \not\models \Exp_A$. 

It may further be of interest to observe that we still have that \[ \initsecreto,\sigma \models K_a \cExp_a \] Agent $a$ is a correct super expert (but not yet $b,c,d$). First, note that after call sequence $ab.cd.ac.bd.ab.ad$, agent $a$ only considers that call sequence possible (in call $ac$ she learns that the second call was $cd$, and in second call $ab$ she learns that the previous call was $bd$). Agent $a$ is a correct superexpert because she has for each other agent two independent observations of the correct value of their secret: $ab$ and $ab$ for the secret of $b$ (obvious), $ac$ and the subsequent $ab$ for the secret of $c$,\footnote{An error for $c$ cannot have been in call $ac$ or call $ab$, as the observed values would then be different, where it is important that $c$ did not call $b$ after call $ac$ and before call $ab$; if the error had been in call $cd$ then $b$ would have passed on the incorrect value to $a$ in second call $ab$ and $a$ would have observed a conflict for $c$ in that call; if the error had been in call $bd$, $a$ would similarly have had a conflict for $c$ in subsequent $ab$.} and $ac$ and $ad$ for the secret of $d$.\footnote{An error for $d$ cannot have been in calls $ac$ or $ad$, as the observed values of $d$ were the same. The error cannot have been in call $cd$, as call $bd$ would then have been correct so that $a$ would have observed conflict in $ad$. It cannot have been in $bd$, as the conflict would then already have appeared in the next $ab$. It cannot have been in $ab$ either, for the same reason.}

Next, consider \[ \tau = ab.ac.ad^d{\!\!}.ab.ac.bc.bd.bd.cd.cd \] After the second call $ac$, agent $a$ incorrectly believes that the secrets are $abc\ov{d}$. After call $bc$ agents $b$ and $c$ still incorrectly believe that the secrets are $abc\ov{d}$, in the first call $bd$ agent $b$ obtains a conflicting value for secret $d$ which is corrected in the second call $bd$, and similarly for $c$ in the subsequent $cd.cd$. We have that \[ I, \tau \models E_A\Exp_A \et \neg E_A\cExp_A \] Not only does $a$ still hold an incorrect secret of agent $d$, but $a$ also believes that $b$ and $c$ hold that incorrect secret of $d$.

Finally, a correct supersuccessful call sequence can be obtained by all agents calling each other twice (and therefore having two independent observations for all secrets, establishing knowledge), thus in $2\binom{n}{2}$ calls ($12$ calls for $n=4$) in a schedule not containing errors: let $\rho = ab.ac.ad.bc.bd.cd$ then we have \[ I, \rho.\rho \models E_A\Exp_A \et E_A\cExp_A \] Any faulty call occurring in such a call sequence seems likely to lead to lengthier sequences, but as long as scheduling of calls is fair, all calls will occur again at some stage and thus any conflict will be resolved. It is unclear if call sequences with faulty calls always take longer to reach correct super success, because agent $a$ correcting a conflict for secret $b$ thus learns that all other values $c$ she holds must be correct, which speeds up the process again.
\end{example}

We close with another example explaining the phenomenon of lucky calls and how this interacts with incorrect values for secrets. A \emph{lucky} call is a call wherein an agent learns that another agent is an expert without calling that agent \cite{hvdetal.lucky:2024}. So it is then not necessary to call that agent in order to find out whether the agent knows all secrets.

\begin{example} \label{example.lucky}
Consider the call sequence $ac.ad.ac.bc.bc.ac$ without transmission errors, and initial secret distribution $a|b|c|d$. After the first three calls, agent $a$ knows that agents $a,c,d$ know the secrets $a,c,d$. In the call $bc$, agents $b$ and $c$ become experts.\footnote{We added another call $bc$ to keep agent $a$ in the dark about the identity of the callers in $bc.bc$. Such an extra call is not needed in the asynchronous version on which the example is based.} In the call $ac$, agent $a$ also becomes an expert and, as $c$ already was an expert in that call, it is easy to see that \emph{agent $a$ also learns that $b$ must be an expert}. Agent $a$ is lucky. However, $a$ does not know if $b$ became an expert by calling $c$ or by calling $d$: she cannot distinguish the actual call sequence from the call sequence $ac.ad.ac.bd.cd.ac$. Agent $a$ therefore also does not know whether $d$ is an expert, as $d$ is not one in the actual call sequence, but is an expert in the other call sequence, that $a$ considers possible.

Now consider the call sequence $ac.ad.ac.b^bd.cd.ac$ with a transmission error made in the call $b^bd$ so that $d$ passes on the incorrect value of $b$ to $c$ and after that $c$ tot $a$. Then after final call $ac$ agents $a$, $c$ and $d$ all incorrectly believe that $b$'s secret is $\ov{b}$ (we have that, respectively, $\ov{b}_a \et \neg b_a$, $\ov{b}_c \et \neg b_c$ and $\ov{b}_d \et \neg b_d$ hold). Agent $a$ now incorrectly believes that $b$ is an expert holding secrets $a\ov{b}cd$, in the sense that $a$ considers possible (for example) gossip state $(a|\ov{b}|c|d, ac.ad.ac.bc.bc.ac)$ and where $a|\ov{b}|c|d, ac.ad.ac.bc.bc.ac \models a_b\et \ov{b}_b \et c_b \et d_b$, whereas actually $b$ is the only correct expert: $a|b|c|d, ac.ad.ac.b^bd.cd.ac \models a_b\et b_b \et c_b \et d_b$!
%
\end{example}

\section{Weaker and stronger call semantics} \label{section.weakerandstronger}

In this section we consider alternative call semantics. First, instead of agents exchanging all the secrets they know, there are more refined message semantics than that, wherein the agents only send or only receive these secrets, or only some but not all of the secrets, or  only one secret per message. In gossip, merely sending is known as `push', merely receiving is known as `pull', and the exchange of secrets of our call semantics is known as `push-pull' \cite{hedetniemietal:1988}. These are standard variations. We will not consider those variations but stick to agents exchanging all the secrets they hold. This comes closer to exchanging `all you know' in order to speed up information dissemination as much as possible, as in full information protocols \cite{MosesT88} and in resolving distributed knowledge \cite{AgotnesW17}. Still, from the perspective of exchanging all that is known, one could say that we hold somewhat of a middle ground, as we use some but not all the information available from the history of prior calls. We use some information, because the agents reason over call sequences containing at most one error: this implicitly means that agents can `count' the number of observations of the value of a secret (namely how often this occurs in a call sequence they consider possible), and can `remember' having corrected a conflicting value for a secret (namely if they only consider call sequences possible wherein they had to self-correct): quotes that are indeed intended to scare, as the logical language does not have such primitives: neither is explicit. We now explore some other options to use the history of calls in the semantics.

We first consider an alternative message semantics wherein agents only store, for each other agent, the information exchanged in the last call with that agent. This is a less expressive semantics $\models^\last$. Such protocols with bounded memory have been considered in distributed computing in \cite{DH08,abs-1803-03042}. Second, we consider a message semantics wherein agents store the entire history of calls and secret holdings of all other agents. This is more expressive semantics $\models^\full$. Those message semantics are known in distributed computing as {\em full-information protocols} \cite{MosesT88}.

\subsection{Call semantics only storing the last call} 

Consider an alternative semantics for gossip wherein: \begin{quote} Each agent $a$ stores for all agents $b \neq a$, for the last call involving $a$ and $b$, the pair $(X,Y)$ consisting of the set $X$ of values of secrets contributed by $a$ to that call, and the set $Y$ of secret values received by $a$ from $b$ in that call. \end{quote} Note it says {\em received by $a$} and not {\em sent by $b$}, as a transmission error may have occurred. In case no call between $a$ and $b$ took place yet, agent $a$ stores its initial secret value $a$ or $\ov{a}$ for herself, and $\emptyset$ for $b$: this is the pair $(\{a\},\emptyset)$ or $(\{\ov{a}\},\emptyset)$. An initial secret distribution $\initsecret$ and a call sequence $\sigma$ thus determine an $n$-tuple that we denote $\last(\initsecret,\sigma)$, with for each agent $a$ a projection $\last(\initsecret,\sigma)_a$. From $\last(\initsecret,\sigma)_a$ we can determine agent $a$'s holding of values denoted $\initsecret[\sigma]_a^\last$, where we do not specify how this is exactly determined. It is not entirely trivial: although each agent stores the last call involving any other agent, the agent does not store the order of these $n-1$ calls. It is {\bf not} known which of these was the actual last call, in which case we could have simply taken the union $X \union Y$ of the pair $(X,Y)$ associated with that call. And there is the issue of errors.

This setup defines a logical semantics denoted $\models^\last$ to distinguish it from our $\models$ semantics. We define $(\initsecret,\sigma) \sim^\last_a (\initsecretp,\tau)$ iff $\initsecret_a = \initsecretp_a$ and $\last(\sigma)_a = \last(\tau)_a$. With that, the semantics for knowledge becomes: \begin{itemize} \item $\initsecret,\sigma \models^\last K_a \phi$ iff $\initsecretp,\tau \models \phi$ for all $\initsecretp$ and $\tau$ such that $(\initsecret,\sigma) \sim^\last_a (\initsecretp,\tau)$. \end{itemize} The other inductive clauses of $\models^\last$, for negation and conjunction, are then as for the satisfaction relation $\models$. What atomic propositions should be associated with this storage of last calls is unclear, and a bit up for grabs. It is at least clear that the set of atomic propositions should now be different. We need more atoms to describe the local state of an agent, not merely $b_a$ for `agent $a$ holds secret value $b$' but instead of that, or additional to that, $b^c_a$ for `agent $a$ received secret value $b$ from $c$' in the last call with $c$; and with, let us say, an initial atom $a^a_a$ as well (and all that for $\ov{b}$ and $\ov{a}$ as well). Agent $a$ would then have a conflict for the secret of $b$ if she holds conflicting values $\ov{b}^c$ as well as $b^d$ for the secret of $b$. And we would need a mechanism to express self-correction in the language. We will not go into all that, or at least not go into all that even further, but instead: (i) give some results for the error-free case of the last-call semantics, and (ii) give examples demonstrating the different results to be expected for the single-error case of the last-call semantics.

\paragraph*{Error-free last-call semantics}
Even for error-free gossip there are differences between the last-call semantics $\models^\last$ and the standard semantics $\models$. Assume a $\models^\last$ semantics that makes atom $b_a$ true if agent $a$ received secret $b$ in the last call from some agent $c$. So, along the line above, $b_a$ is true iff there is a $c$ such that $b^c_a$ is true. We recall that the initial secret distribution $\initsecreto$ is a stand-in for the (unique) initial secret distribution in error-free gossip.

The following can be shown. First, success corresponds for the last-call and for the standard semantics. This is not surprising. Second, an agent may be a super expert for the standard semantics but not for the last call semantics: this is because an agent may be \emph{lucky} and learn that another agent is an expert without calling that agent \cite{hvdetal.lucky:2024}. This is more remarkable. Third, a last-call super expert is also a standard super expert. Fourth, super success for the last-call semantics implies super success for the standard semantics. 
\begin{proposition}\label{prop.luckrunsout} \
\begin{enumerate}
\item $\initsecreto,\sigma\models \Exp_A$ iff $\initsecreto,\sigma\models^\last \Exp_A$
\item $\initsecreto,\sigma\models K_a\Exp_A$ does not imply $\initsecreto,\sigma\models^\last K_a \Exp_A$
\item $\initsecreto,\sigma\models^\last K_a\Exp_A$ implies $\initsecreto,\sigma\models K_a \Exp_A$
\item $\initsecreto,\sigma\models^\last E_A \Exp_A$ implies $\initsecreto,\sigma\models E_A\Exp_A$
\end{enumerate}
\end{proposition}
\begin{proof} \
\begin{enumerate}
\item $(\Imp)$: If $a$ holds secret $b$, then (for $b \neq a$) $a$ must have received secret $b$ from some agent $c$, not necessarily for the first time in the last call with $c$, but then $c$ would still have sent $b$ to $a$ in the last call between $a$ and $c$. This holds for all $a$ and $b$. \\ $(\Pmi)$: If $\initsecreto,\sigma\models^\last \Exp_A$ then some agent $c$ sent $b$ to $a$ in the last call between $a$ and $c$. Therefore $a$ holds the secret $b$. 
\item Consider the call sequence $ac.ad.ac.bc.bc.ac$ without transmission errors from Example~\ref{example.lucky}, wherein agent $a$ is lucky about agent $b$ in final call $ac$. At this stage agent $a$ knows that $a,b,c$ are experts. Now extend this call sequence with call $ad$. Then $a$ is a super expert: $K_a\Exp_A$ is now true. But agent $a$ never called $b$ and therefore stores $(\{a\},\emptyset)$ for the last call involving $b$. Therefore, although $\initsecreto,\sigma\models^\last \Exp_b$, we still have $\initsecreto,\sigma\not\models^\last K_a \Exp_b$ and therefore $\initsecreto,\sigma\not\models^\last K_a \Exp_A$.
\item 
Let pair $(X,Y)$ be associated with the last call between $a$ and $c$. If $X \union Y = A$, then $a$ knows that that $c$ is an expert after the call, and knows that in either semantics, that is, $\initsecreto,\sigma\models^\last K_a\Exp_c$ and $\initsecreto,\sigma\models K_a \Exp_c$. From assumption $\initsecreto,\sigma\models^\last K_a\Exp_A$ it follows that $\initsecreto,\sigma\models^\last K_a\Exp_c$ holds for all $c \in A$. Therefore also $\initsecreto,\sigma\models K_a \Exp_A$. 
\item The fourth item is a consequence of the third.
\end{enumerate}
\end{proof}
Concerning the fourth item, we conjecture that also:  
\[ \initsecreto,\sigma\models E_A\Exp_A \text{ implies } \initsecreto,\sigma\models^\last E_A \Exp_A \] With an asynchronous call semantics, even when an agent $a$ is lucky about $b$, in order to become super experts, agent $b$ still has to call $a$ in order to learn that $a$ is an expert. So in the end, {\em all agents have to be involved in a call to each other after which they are both expert} to reach the super expert goal \cite[Lemma 34]{hvdetal.lucky:2024}. If it were shown that this requirement also holds for a synchronous semantics, then the implication would be established. 

\paragraph*{Last-call semantics with errors}
However, let us now consider some scenarios involving the last-call semantics $\models^\last$ and faulty calls. A modal logical issue with the last-call semantics is what `knowing the secret' now means, as the accessibility relations of the last-call semantics, and thus the notion of knowledge, are completely determined by the local states of the agents. In the $\models^\last$ semantics `$a$ knows secret $b$' therefore necessarily corresponds to a local state value of agent $a$, whereas in our semantics `$a$ knows secret $b$' corresponds to $a$ knowing a local state value of agent $b$. Obvious candidates for `$a$ knows secret $b$' in the last-call semantics are: (i) the value of $b$ received by $a$ in the last call with $b$, or (ii) the same values of $b$ received by $a$ from all agents $c$. And all that in their last call with $a$. Unfortunately, for both interpretations, where we again write $\Kv_a b$ for `$a$ knows secret $b$', and where $\initsecret$ is an initial secret distribution:
\begin{observation} \label{observation.one}
$\initsecret,\sigma\models^\last \Kv_a b$ does not imply $\initsecret,\sigma\models \Kv_a b$.
\end{observation}
For reading (i) it is obvious, as in the single call from $b$ to $a$ there could be a transmission error for secret $b$. But for reading (ii) the implication also fails. Consider four agents, secret distribution $\initsecreto = a|b|c|d$, and call sequence $b^ba.ac.ad.ca.da$. Then $a$ received $\ov{b}$ from all three other agents in their last call. But $\ov{b}$ is still the incorrect value.

In the last-call semantics we can still imagine $a$ eventually getting to know the correct value of a secret by holding conflicting values $b$ and $\ov{b}$ and receiving $b$ from another agent (and self-correcting on a majority of two $b$'s), or even by already holding a value $b$ and receiving $b$ from another agent (a majority of two $b$'s without need to self-correct). But even after $a$ receives value $b$ from agent $b$ a thousand times, when in the next call with $b$ she after all receives $\ov{b}$ then she will have a conflict, as she only compares this to the value $b$ received in the final of those one thousand calls. Therefore:
\begin{observation} \label{observation.qwerty}
$\initsecret,\sigma\models^\last \Kv_a b$ and $\sigma\sqsubseteq\tau$ do not imply $\initsecret,\sigma\models^\last \Kv_a b$.
\end{observation}

In the standard semantics, there is a difference between correct belief, that is not knowledge of secrets, and justified correct belief, that counts as knowledge of secrets. But not in the last-call semantics. Now, correct belief is by definition known correct belief, as it is a feature of the local state.

Such negative results suggest that it would be very hard to get correct termination. But strangely enough this is not the case. Because (see the contrast with Observation~\ref{observation.one}):
\begin{observation}
$\initsecret,\sigma\models^\last \Et_{a,b \in A} \Kv_a b$ implies $\initsecret,\sigma\models \Et_{a,b \in A} \Kv_a b$.
\end{observation}
In other words:
\begin{observation}
$\initsecret,\sigma\models^\last E_A\cExp_A$ implies $\initsecret,\sigma\models E_A\cExp_A$.
\end{observation}
This is because when all agents agree on all values, also locally for all values they received from other agents, any conflicts must have been resolved or have never occurred. The last-call semantics is the precisely minimal and adequate one to obtain such correct super success, which is exactly why we wished to present it as an alternative. This is relevant because it may well be that super success is the strongest goal that can be obtained for synchronous gossip (this has only been proved for asynchronous gossip \cite{hvdetal.lucky:2024}). 

The result does not hold in the other direction that $\initsecret,\sigma\models E_A\cExp_A$ implies $\initsecret,\sigma\models^\last E_A\cExp_A$. Super success is unstable in the last-call semantics, because knowledge $\Kv_a b$ is unstable (Observation~\ref{observation.qwerty}).

Another plus of the last-call semantics is that it is {\em first-correct supersuccessful}, unlike the standard semantics (Proposition~\ref{anyresults}).
\begin{observation}
$\models^\last E_A\Exp_A \imp E_A\cExp_A$.
\end{observation}
This is a consequence of our choice of defining knowledge of a secret for the last-call semantics. An agent cannot agree on the values of all secrets for all agents unless they are all correct, otherwise there must have been a conflict. So in a way the last-call semantics is `only-correct supersuccessful'.

Continuing on that theme, despite being first-correct supersuccessful, the last-call semantics in not first-correct successful:
\begin{observation}
$\not\models^\last \Exp_A \imp \cExp_A$.
\end{observation}
A counterexample for four agents is $ab.ac.bc.ad^d{\!\!}.ab.ac$, after which all agents hold values for all secrets and without conflict, but agents $a,b,c$ hold the incorrect value of secret $d$. 

\subsection{Full information protocol}

In an alternative $\models^\full$ semantics for calls, agents do not merely exchange sets of values of secrets but they exchange full information, that is, they also exchange trees (or dags, see below) of subsequences of calls, namely of their own previous calls, but also of the previous calls of other agents having called them, and so on. Full information protocols are well-known from distributed computing \cite{MosesT88}. We first compare full-information semantics $\models^\full$ with the standard semantics $\models$ for correct gossip, and then repeat the exercise, more tentatively, for gossip with errors.

\paragraph*{Error-free full-information semantics}
In full-information semantics, given a gossip state $(\initsecreto,\sigma)$, in a call $ab$ the agents $a$ and $b$ exchange in a call not only their sets of secrets $\initsecreto[\sigma]_a$ respectively $\initsecreto[\sigma]_b$ but also their {\em full views} $\fv_a(\sigma)$ respectively $\fv_b(\sigma)$. The inductive definition is as follows for synchrony, where $b,c \neq a$. By identifying identical subtrees, it can be said to construct a \emph{dag}, a directed acyclic graph. The $\bullet$ symbol represents the call not involving $a$ (but $a$'s awareness of the global clock). We follow the presentation in \cite[Chapter 2]{hvd:2026}.
\[ \begin{array}{lllr}
\fv_a(\epsilon) & := & \epsilon \\
\fv_a(\sigma.bc) & := & \fv_a(\sigma).\bullet  \\
\fv_a(\sigma.ab) & := & (\fv_a(\sigma),\fv_b(\sigma)).ab \\
\fv_a(\sigma.ba) & := & (\fv_b(\sigma),\fv_a(\sigma)).ba 
\end{array}\]
If we now define (synchronous) full view observation relations just as the observation relations from Definition~\ref{def.observationrelation}, and restricted to the error-free call semantics, the only different clause is the one for agents involved in the call. We put the old and the new one next to each other, to highlight the difference (the clause for $ba$ is similar): 
\[\begin{array}{lclcl} (\initsecreto,\sigma.ab) \sim^\full_a (\initsecreto,\tau.ab) &\text{ \quad iff \quad }& (\initsecreto,\sigma) \sim^\full_a (\initsecreto,\tau) &\text{ and }& (\initsecreto,\sigma) \sim^\full_b (\initsecreto,\tau) \\ (\initsecreto,\sigma.ab) \sim_a (\initsecreto,\tau.ab) &\text{ iff }& (\initsecreto,\sigma) \sim_a (\initsecreto,\tau) &\text{ and }& \initsecreto[\sigma]_b = \initsecreto[\tau]_b \end{array}\]
One can show \cite{hvd:2026} that $(\initsecreto,\sigma) \sim^\full_a (\initsecreto,\tau)$ implies $\fv_a(\sigma)=\fv_a(\tau)$, and that $(\initsecreto,\sigma) \sim^\full_a (\initsecreto,\tau)$ implies $\initsecreto[\sigma]_a = \initsecreto[\tau]_a$.

The full-information semantics is therefore at least as strong as the standard semantics in the sense that any positive information about facts obtained with the standard semantics, such as $b_a$, $K_c b_a$, $K_c b_b$ and $\Kv_c b$, is also obtained with full information. For such formulas $\phi$ we have that $\initsecreto,\sigma \models \phi$ implies $\initsecreto,\sigma \models^\full \phi$. But it is also really stronger, as we can reach arbitrary higher-order epistemic termination goals. One can show, first, that the goal of super success is reached by the same call sequence in the full-information semantics; second, that the goal of super success can be reached faster in the full-information semantics; and third, that there are higher order goals of mutual knowledge of all secrets than can be reached in the full-information semantics but not in the standard semantics. (A way to describe the third result that contrasts better with the previous two, is to state that \[ 3. \ \initsecreto,\sigma\models^\full E_A E_A E_A\Exp_A \text{ does not imply } \initsecreto,\sigma\models E_A E_A E_A \Exp_A. \] However, the right-hand side there does not imply unsatisfiability, which is stronger.)
\begin{proposition} \label{proposition.fullinfo} \ 
\begin{enumerate}
\item $\initsecreto,\sigma\models E_A\Exp_A$ implies $\initsecreto,\sigma\models^\full E_A \Exp_A$;
\item $\initsecreto,\sigma\models^\full E_A \Exp_A$ does not imply $\initsecreto,\sigma\models E_A\Exp_A$;
\item $E_A E_A E_A\Exp_A$ is $\models^\full$ satisfiable but is not $\models$ satisfiable.
\end{enumerate}
\end{proposition}
\begin{proof} \ 

\begin{enumerate}
\item Induction on the length of call sequences proves that $\initsecreto,\sigma \models K_a b_c$ implies $\initsecreto,\sigma \models^\full K_a b_c$. Here we use that exchange of secrets is part of the standard call semantics but also of the full-information protocols, so that $\initsecreto,\sigma.ab \models c_a$ if $\initsecreto,\sigma \models c_a$ or $\initsecreto,\sigma \models c_b$ just as well as $\initsecreto,\sigma.ab \models^\full c_a$ if $\initsecreto,\sigma \models^\full c_a$ or $\initsecreto,\sigma \models^\full c_b$. We also use that all variables are local: $b_a \eq K_a b_a$ is valid in either semantics. Now note that $K_a b_a$ is a conjunct of $E_A \Exp_A = \Et_{a,b} K_a b_a$. 
\item With full-information gossip protocols one can get knowledge faster than with standard gossip protocols. An example for four agents $a,b,c,d$ is the eight-call sequence $\sigma = ab.cd.ac.bd.ab.ad.bc.cd$ that satisfies $\initsecreto,\sigma \models E_A\Exp_A$, where eight calls is the optimal $n-2+\binom{n}{2}$ for $n=4$. We can remove calls $ad$ and $bc$ from this sequence and still obtain that goal in the full semantics but therefore not standardly: \[ \begin{array}{l}
\initsecreto,ab.cd.ac.bd.ab.bc \models^\full E_A\Exp_A \\ \initsecreto,ab.cd.ac.bd.ab.bc \not\models E_A\Exp_A \end{array}\] The reason is, that in the full information protocol, in call $ab$ prior to call $ad$, agent $a$ by also sending $b$ her full view (as also implied by the definition of $\sim^\full_a$) effectively informs agent $b$ that she learnt in call $ac$ that agents $a$ and $c$ are experts (`all she knows') and that in that same call agent $b$ similarly informs agent $a$ that agents $b$ and $d$ are experts. Then, also in the same way, in call $cd$ agent $c$ informs agent $d$ that $c$ and $a$ are experts and agent $d$ informs agent $c$ that $b$ and $d$ are experts. Therefore call $ad$ has no longer to take place in order for $a$ to learn that $d$ is an expert and for $d$ to learn that $a$ is an expert, and similarly for call $bc$. (This example and similar examples are found in \cite{HerzigM17,CooperHMMR19,hvd:2026}.) 

\item The first item showed that $E_A\Exp_A$ is $\models^\full$ satisfiable. Assume a $\models^\full$ super-successful call sequence $\sigma$, that is, $\initsecreto,\sigma\models^\full E_A \Exp_A$. Then all agents $a$ are super experts so that $\initsecreto,\sigma\models^\full K_a \Exp_A$. We first show that $\initsecreto,\sigma \models^\full K_a \Exp_A$ implies $\initsecreto,\sigma.ab \models^\full K_b K_a \Exp_A$. 

\medskip

\noindent $
\initsecreto,\sigma \models^\full K_a \Exp_A \\
\Eq \\
\initsecreto,\sigma \models^\full K_a K_a \Exp_A \\
\Eq \\
\initsecreto,\tau \models^\full K_a\Exp_A \text{ for all } (\initsecreto,\tau) \sim^\full_a (\initsecreto,\sigma) \\
\Imp \\
\initsecreto,\tau \models^\full K_a\Exp_A \text{ for all } (\initsecreto,\tau) \sim^\full_a (\initsecreto,\sigma) \text{ and } (\initsecreto,\tau) \sim^\full_b \initsecreto,\sigma) \\
\Imp \hfill \text{preservation of factual knowledge} \\
\initsecreto,\tau.ab \models^\full K_a\Exp_a \text{ for all } (\initsecreto,\tau) \sim^\full_a (\initsecreto,\sigma) \text{ and } (\initsecreto,\tau) \sim^\full_b (\initsecreto,\sigma) \\
\Eq \\
\initsecreto,\tau.ab \models^\full K_a\Exp_a \text{ for all } (\initsecreto,\tau.ab) \sim^\full_b (\initsecreto,\sigma.ab) \\
\Eq \\
\initsecreto,\sigma.ab \models^\full K_b K_a\Exp_A
$

\medskip

\noindent As we can do thus reach $K_b K_a\Exp_A$ for all $a,b \in A$, we have that $\sigma.\tau \models^\full E_AE_A \Exp_A$ where $\tau$ is the sequence of all $\binom{n}{2}$ calls. Not only knowledge of secrets is preserved after call sequence extension, but also higher-order knowledge, so the proof above can be adapted to show that $\initsecreto,\sigma.ab \models^\full K_b K_a \Exp_A$ implies $\initsecreto,\sigma.ab.bc \models^\full K_c K_b K_a \Exp_A$.\footnote{One could imagine a proof for $K_a \phi^{\mathsf{pos}}$ for any $\phi^{\mathsf{pos}}$ in the fragment $b_a \mid \phi \vel \phi \mid \phi \et \phi \mid K_a \phi$ which is almost like the positive fragment corresponding to the universal fragment in first-order logic, except for the absence of basic clause $\neg b_a$.} Combining the two we then get that $\initsecreto,\sigma.\tau.\tau \models^\full E_A E_A E_A \Exp_A$. Much faster schedules exist, but we only care about satisfiability and not about optimality here.

It remains to show that $E_A E_A E_A \Exp_A$ is unsatisfiable with the standard call semantics. Now for asynchrony our life would be been simpler: $E_A E_A\Exp_A$ is unsatisfiable in the {\em asynchronous} $\models$ semantics \cite{hvdetal.lucky:2024}, so it would already have sufficed to show that $E_A E_A\Exp_A$ is satisfiable in the $\models^\full$ semantics. But it is unknown whether $E_A E_A\Exp_A$ is unsatisfiable in the synchronous $\models$ semantics (although it seems likely that it is unsatisfiable). In the asynchronous semantics one can show that, even when two agents become super experts in the same call, they both consider it possible that the other agent did not become a superexpert. We do not know if there is information leakage with synchrony such that agents that become super experts could learn this from one another. But we have another trick up our sleave. First, the synchronous and asynchronous standard semantics of error-free gossip coincide in the fact that once all agents are super experts, any further calls are not informative. Agent $a$ can now predict the information exchange with any other agent $b$: only the set of all secrets will be exchanged. What you can predict is not informative. Second, given synchrony, for all subsequent calls not involving herself, she knows (because she is a super expert) that in all such cases only agents both already knowing all the secrets will call: not informative. Therefore, even if $E_AE_A\Exp_A$ where satisfiable, yet higher-order termination goals are unsatisfiable: $E_AE_A E_A\Exp_A$ is unsatisfiable in the synchronous $\models$ semantics.
\end{enumerate}
\end{proof}
We conjecture that $E_A E_A \Exp_A$ is also unsatisfiable for synchronous gossip, just as for asynchronous gossip, in which case we would also have that $\initsecreto,\sigma\models^\full E_A E_A\Exp_A$ does not imply $\initsecreto,\sigma\models E_A E_A \Exp_A$.

\begin{example}
Given four agents $a,b,c,d$ and call sequence $\sigma = ab.cd.ac.bd.ab.ad.bd.cd$ we have that $\initsecreto,\sigma\models E_A\Exp_A$. Then also $\initsecreto,\sigma\models^\full E_A\Exp_A$. Extend this call sequence with all six calls between two agents $\tau = ab.ac.ad.bc.bd.cd$. Then $\initsecreto,\sigma.\tau\models^\full E_A E_A\Exp_A$. Now extend that call sequence once more with all six calls $\tau$ between two agents. Then $\initsecreto,\sigma.\tau\models^\full E_A E_A E_A\Exp_A$. Note that this is $\bigO(n^2)$ for $n$ agents (we made $\binom{n}{2}$ calls, three times). Faster schedules to reach this goal, of $\bigO(n)$, are given in \cite{HerzigM17}.
\end{example}

\paragraph*{Full-information semantics with errors}

In full-information semantics with errors we need to take arbitrary initial secret distributions into account in the observation relation and we thus get that $(\initsecret,\sigma) \sim^\full_a (\initsecretp,\tau)$ iff $\fv_a(\initsecret,\sigma)=\fv_a(\initsecretp,\tau)$, and furthermore there may be transmission errors in the communicated secret values. Instead of providing a formal definition, let us sketch some issues and consequences. We can no longer compute the current holding of an agent from its full view and the initially held secret values. That agents share their sets of (values of) secrets in each call is now essential, as there may be transmission errors. Assuming that the semantics of a call $ab$ is the same in $\models$ and $\models^\full$, and an appropriate adjustment of the observation relation $\sim^\full_a$ for faulty calls, we expect that we still get (where $\initsecret$ is an arbitrary initial secret distribution):
\begin{conjecture} $\initsecret,\sigma\models E_A\Exp_A$ implies $\initsecret,\sigma\models^\full E_A \Exp_A$
\end{conjecture}
A piece of good news is that when $a$ obtains a conflict for a secret $b$ and resolves this by for example calling $b$, thus obtaining a majority of correct values of $b$'s secret, in any next call to anyone, including $b$, she will not only communicate the correct value of $b$ to anyone, but also communicate her knowledge $\Kv_a b$ of the correct value, and so on, including higher-order knowledge of secrets. It therefore seems that also for gossip with errors:
\begin{conjecture}  $\initsecret,\sigma\models^\full E_A E_A E_A\Exp_A$ does not imply $\initsecret,\sigma\models E_A E_A E_A \Exp_A$
\end{conjecture}

\section{Discussion and conclusion}

There are some clear topics for further research. First, we presented a synchronous semantics of gossip with errors, but we would like to show similar results for an asynchronous semantics of gossip of errors. This poses some additional technical complications. Second, we restricted ourselves to one transmission error because the restriction came with sufficient complications and already results in a modal epistemic setting. But one should obviously think of multiple errors. Third, instead of transmission errors we also wish to investigate faulty or so-called \emph{Byzantine} agents. In the first subsection below we discuss the issues with asynchrony, the generalization of our results from one transmission error to a bound of $f$ transmission errors, and the formalization of one or $f$ faulty agents.

In this work holding a secret and knowing a secret are really different epistemic notions. Other epistemic notions may also come into question. In the second subsection below we succinctly present connections to some other epistemic notions: belief, hope, protocol-dependent knowledge, and local state knowledge. 

The final third subsection contains a short conclusion.

\subsection{Asynchrony, multiple errors, and faulty agents}

We succinctly discuss asynchrony, the generalization to gossip protocol executions with a bound of $f$ transmission errors, and how to address gossip with at most $f$ faulty agents.

\paragraph*{Asynchrony} If we change the clauses of the observation relation for calls not involving the distinguishing agent $a$, and leave everything else in the semantics the same --- in the semantics of calls, the observation relation, and the satisfaction relation --- we get an asynchronous semantics of gossip. Let us show the novel clauses for the observation relation; where $cd^{\kappa}$ is any call involving $c$ and $d$, faulty or not. 
\[\begin{array}{lll}
(\initsecret,\sigma.cd^\kappa) \sim_a (\initsecretp,\tau.fg^\kappa) & \text{iff} & (\initsecret,\sigma) \sim_a (\initsecretp,\tau) \hspace{2cm} \hfill \text{synchronous} \\
(\initsecret,\sigma.cd^\kappa) \sim_a (\initsecretp,\tau) & \text{iff} & (\initsecret,\sigma) \sim_a (\initsecretp,\tau)  \hfill \text{asynchronous}
\end{array}\]
In the synchronous semantics a $\sim_a$-equivalence class consists of a finite number of call sequences of the same length. But in the asynchronous semantics a $\sim_a$-equivalence class always consists of an infinite number of call sequences, because agent $a$ cannot distinguish the actual call sequence from a call sequence also containing an arbitrarily large number of subsequent calls not involving it (and also not from many other call sequences having such calls prior to $a$'s last call). In order to investigate whether these asynchronous semantics of gossip with errors are also well-defined, we therefore cannot use the previous partial order between pairs of call sequences and formulas, that was based on the equal length of indistinguishable call sequences. For example, consider $K_a \phi$. The semantics of knowledge (that is the same for synchrony and asynchrony) gives us that: $\initsecret,\sigma \models K_a \phi$ iff $\initsecretp,\tau \models \phi$ for all $(\initsecretp,\tau) \sim_a (\initsecret,\sigma)$. However, this time round we do not have that $(\tau,\phi) < (\sigma, K_a \phi)$, because the required $|\tau|=|\sigma|$ to establish this need not be the case. Call sequence $\tau$ may be arbitrarily long, and in particular longer than $\sigma$. 

We intend to attack this problem by defining an alternative partial order based on so-called {\em direct knowledge of secrets} \cite{logicofgossiping:2020} (see also \cite{hvd:2026}). If in a call from $a$ to $b$ either $a$ learns a new secret from $b$, or $b$ learns a new secret from $a$ we say that the direct knowledge of secrets of $a$ respectively of $b$ increases. As for $n$ agents there are $n^2$ secrets to know there is an upper bound on the direct knowledge of secrets. One can then show that, given the restricted information exchange in gossip where agents only exchange sets of secrets, the maximum length of call sequences in which direct knowledge of secrets changes is $2n^3$  \cite{logicofgossiping:2020}, and this is then used to show that each infinite $\sim_a$-equivalence class is equally informative as a finite $\sim_a$-equivalence class where the maximum length of call sequences is $2n^3$. Any other call is then redundant (bisimulation invariant). For a logical language as ours but with also call modalities $[ab]\phi$, this is then sufficient to completely axiomatize asynchronous knowledge. It seems such results are reproducable in an asynchronous logic of gossip with errors, where of course (in case of at most one transmission error) the number of directly knowable secrets is twice that of error-free gossip. This does not seem to affect the order of magnitude of the above $2n^3$ but only the constant factor. One could then envisage a partial order not based on {\em the length $|\sigma|$ of call sequences $\sigma$} but on {\em the number $\sharp\sigma$ of directly known secrets after $\sigma$}. Surely there are several snakes in the grass lying in wait for the intrepid investigator.

\paragraph*{Multiple transmission errors}

We can lift the restriction of one transmission error per protocol execution to $f$ transmission errors per protocol execution. One call may now contain multiple transmission errors, and in both directions of the exchange. We need to choose between having for each agent $f$ \emph{different} incorrect values, on the assumption that all transmission errors are different, or having $f$ \emph{occurrences} of one incorrect value. The latter seems more suitable for faulty agents than for transmission errors, as faulty agents can repeatedly and intentionally send the same faulty message, whereas all transmission errors are presumably different. We recall that the two values $a$ and $\ov{a}$ of the secret  of $a$ `really' are, respectively, pairs $(a,1)$ and $(a,0)$ of the set of secret values $A \times \{0,1\}$. By enlarging this to $A \times \{0,\dots,f\}$ we get our set of at most $f$ incorrect values per agent ($f+1$ values per agent, as it includes one correct value), denoted as $a^0,\dots,a^f$. However, in a different usage, we can also see these $a^0,\dots,a^f$ are named occurrences of the same error. Already with merely two values but allowing $f$ occurrences of an incorrect value, we can generalize the epistemic semantics with agents reasoning about sets of indistinguishable call sequences: instead of a majority of two in order to decide that a value is correct (that is, one more than the maximum of one error), we now need a majority of $f+1$, and the worst case in order to be able to obtain that therefore consists of $2f+1$ calls, a well-known figure in distributed computing \cite{abs-2106-11499}. In the presence of multiple errors or multiple names of an error, knowing the secret would now be defined as $\Kv_a b := K_a b^0_b \vel \dots \vel K_a b^f_b$ and the semantics of a call $ab$ would need clauses $*$ removing all values $b^i$ among $i=0,\dots,f$ known to be incorrect from $a$'s holding on condition that $\Kv_a b$ and therefore exactly one such $K_a b^j_b$ for $j \neq i$ holds. 

The results for one error seem easily reproduced for the case of $f$ errors. Beyond that, one can consider {\em at most $f$ errors per time frame of $y$ calls} (which seems a more realistic assumption), or {\em at most $f$ errors per agent} (which seems to go easier on the formalization because we have a distributed system). Both directions seem promising to pursue.

Instead of upper bounds on faults one could instead have a probability on the occurrence of faults. But that would call for an entirely different approach, because knowledge as complete certainty in a modal logic like ours can then never be obtained, given the non-zero chance of an error \cite{halpern:2003}. So that does not seem to be a promising direction of research to pursue in a modal logic.

\paragraph*{Faulty agents}

For a bound of $f$ faulty (Byzantine) agents the expert goal of the gossip protocol should become: for all correct agents to hold the secrets of the correct agents, and similarly for the super expert goal. We recall that \cite{BergG20} also requires the correct agents to identify the incorrect agents, which is a stronger goal that seems hard to satisfy in a fully distributed setting.

This subject is also deferred to future research, where we restrict ourselves here to a typical benchmark example for faulty agents, employing the same language and semantics. Given $B \subseteq A$, a {\em $B$ expert} is an agent who holds all secrets in the set $B$, and a {\em $B$ super expert} is an agent who knows that all agents in $B$ hold all secrets in $B$.

As we use the same semantics, in particular we use the same observation relation, encoding that agents may not know they are faulty. Such faulty agents are therefore rather byzantine, in the sense of being randomly incorrect, than malicious or otherwise intentionally incorrect. However, it does at least not rule out that agents are always incorrect, as in the following example.

\begin{example}
Given four agents $a,b,c,d$ and initial secret distribution $\initsecreto = a|b|c|d$, assume $d$ is faulty and that $d$ always sends the incorrect value $\ov{d}$ of its secret. Then after call sequence $ab.cd^d{\!\!}.ac.bd^d$, agents $a,b,c$ correctly know the secrets of $a,b,c$, that is, those in $B=A{\setminus}\{d\}=\{a,b,c\}$ are correct $B$ experts. Now consider the extended call sequence $\sigma = ab.cd^d{\!\!}.ac.bd^d{\!\!}.ab.ad^d{\!\!}.bc.cd^d$. We would like to achieve then that all the (correct) agents in $B$ are $B$ (correct) super experts. But, just as in Example~\ref{FIVE}, this is not the case. The agents do not even know they are correct, as for example the last call could also have been $cd$ instead of $cd^d$ thus giving agent $c$ a conflict for secret $d$; in other words, non-involved agents can not (and can never) rule out the possibility of a conflict in a call in which they are not involved. So $\initsecret,\sigma \not\models E_B\Exp_B$ and a fortiori also $\initsecret,\sigma\not\models E_B\cExp_B$.
\end{example}

\subsection{Epistemic perspectives}

\paragraph*{Belief} We would like to say that values correct and incorrect are \emph{believed} when not known, as there are so few errors. Such belief carries weight in our formalization, as we distinguish the (super) expert goal wherein we do not resolve such beliefs from the correct (super) expert goal wherein all is known. If $b_a \et \neg \ov{b}_a$ is true we say that {\em agent $a$ believes that the secret of agent $b$ is $b$}, and if $\neg b_a\et\ov{b}_a$ is true then {\em agent $a$ believes that the secret of agent $b$ is $\ov{b}$}. We cannot merely identify the truth of local atoms $b_a$ resp.\ $\ov{b}_a$ with this belief, as they also occur in the context $\ov{b}_a\et b_a$ wherein $a$ has a conflict for secret $b$. We can make this belief formal and define $B_a b  :=  b_a \et \neg \ov{b}_a$ and $B_a \ov{b} := \neg b_a \et \ov{b}_a$. We now have that that $\initsecret,\sigma \models B_a b$ iff $\initsecret[\sigma]_a^{\inter b} = \{b\}$ and that $\initsecret,\sigma \models B_a \ov{b}$ iff $\initsecret[\sigma]_a^{\inter b} = \{\ov{b}\}$. 

Because $\Kv_a b$ as disjunct $K_a b_b$ implies $B_a b$, and $\Kv_a b$ as disjunct $K_a \ov{b}_b$ implies $B_a \ov{b}$, `knowledge implies belief' as expected. And belief does not imply knowledge, which is also expected. Even correct belief does not imply knowledge which might have been unexpected: we recall (Section~\ref{section.syntaxsemantics}) that correct belief $b_b \et b_a \et \neg\ov{b}_a$ does not imply knowledge $\Kv_a b$, simply because correct belief $b_b \et b_a \et \neg\ov{b}_a$ does not imply known correct belief $K_a (b_b \et b_a \et \neg\ov{b}_a)$. Correct belief has to be \emph{justified} by the epistemic semantics to become known. Only that amounts to knowing the secret: $\Kv_a b$ is equivalent to $K_a (b_b \et b_a \et \neg\ov{b}_a) \vel K_a (\ov{b}_b \et \neg b_a \et \ov{b}_a)$ (Corollary~\ref{cor.knowledgecorrectbelief2}).

Such belief $B_a b$ is not a modality but a boolean abbreviation. For exampe, we cannot write $\neg\phi\et B_a\phi$ for arbitrary formulas $\phi$ to express that belief is incorrect.

\paragraph*{Hope}
Another epistemic notion to consider is the one known in the literature as \emph{hope} \cite{abs-2106-11499,hvdetal.aiml:2022}. An agent \emph{hopes} $\phi$, if $\phi$ is known conditionally on the agent being \emph{correct}, definable as $H_a \phi := K_a (\mathit{correct}_a \imp \phi)$. Such hope modalities are interpreted in Kripke models for partial equivalence relations (symmetric and transitive relations, with a corresponding modal logic {\bf KB4}), where in the faulty part of the model the agent's relation is empty. Now in our setting, agent $a$ holding conflicting values for a secret $b$ is like $a$ having contradictory values for secret $b$, so it is like $a$ knowing a contradiction (having an empty relation). But without actually being contradictory: to be interpreted as belief instead of conflict we merely require different values of secrets to be mutually exclusive.

\paragraph*{Known protocols}
We assumed that call sequences contain at most one faulty call and  transmission error. This reduces the set of possible secret distributions when extending call sequences, because extensions that contain more than one faulty call are not allowed and can therefore not be indistinguishable for the agent. Or we assume the same, but with more faulty calls and at most $f$ errors. Alternatively we could allow arbitrary call sequences, but impose restrictions on the observation relation and thus on indistinguishable call sequences, that are explicit in the logical language. A good candidate for such a logic-based restriction is the so-called {\em protocol-dependent knowledge} of \cite{DitmarschGKP19,gattinger.phd:2018}: instead of the usual knowledge $K_a \phi$ we now have knowledge $K^\prot_a\phi$ which means that agent $a$ knows that $\phi$ on the assumption that gossip protocol $\prot$ is common knowledge, where we can think of protocol $\prot$ as the combined call conditions $\prot_{ab}$ for all calls $ab$, such as `$a$ does not know the secret of $b$' or `$a$ considers it possible that $b$ does not know her secret'. In this particular case of gossip with errors we could therefore have knowledge $K^{\leq f}_a \phi$ where $f$ becomes false when the call sequences contains more than $f$ errors. (Formulas may also contain epistemic modalities $K^\prot$ and $K^{\prot'}$ for different protocols.) We would thus get the full observation model for the protocol $\ANY$ (including valuations for multiple values of secrets) but where the observation relations are partial equivalence relations that are empty when the `protocol is violated', that is: when call sequences are extended so they contain more than $f$ errors. Such $K^{\leq f}_a$ knowledge modalities are {\bf KB4}, not {\bf S5}. The observation model defined in Section~\ref{section.syntaxsemantics} is simply the restriction of this full observation model to the part with non-empty relations.

\paragraph*{Local state}
In a distributed system it is customary that an agent locally stores all the information it receives in messages from other agents. In contrast, in an (propositional) epistemic semantics some information is stored locally as the value of atomic propositions, in our case local atoms $b_a$ (and $\ov{b}$) for $a$ holding secret value $b$ ($\ov{b}$), and other information is stored indirectly in the Kripke model, in our case: the observation model.  As a consequence, what is known as a local state in distributed computing is far more extensive than the valuation of local atoms for a given agent in a gossip state in the observation model (that is, in a world in a Kripke model). However, we can relate the two perspectives in a precise way.

First: An agent $a$ gathers the following information during the execution of a call sequence: for each call $ab$ involving $a$ (and for no other calls), agents $a$ contributes a set of secret values $X$ and receives from $b$ a set of secret values $Y$ (which include possibly incorrectly transmitted values). The {\em local state of agent $a$} given an initial secret distribution $\initsecret$ and after call sequence $\sigma$ where $a$ is involved in each call of subsequence $\tau$ is therefore: \begin{quote} The set of secrets held by $a$ before each call in $\tau$ and the set of secrets received by $a$ from $b$ in that call. \hfill \emph{local state} (i) \end{quote} A \emph{global state} is then an $n$-tuple of local states. 

Second: In our Kripke semantics the correspondent for agent $a$'s local state is:  \begin{quote} The pair consisting of the observation model and the $\sim_a$ equivalence class in that model containing the actual gossip state $(\initsecret,\sigma)$, as well as agent $a$'s holding $\initsecret[\sigma]_a$ of local atoms. \hfill \emph{local state} (ii) \end{quote} In multi-{\bf S5} Kripke models (with equivalence relations) the structure of the model is common knowledge and so also known to agent $a$, whereas agent $a$'s knowledge is determined by her equivalence class. Again, a global state is the $n$-tuple of local states of which the intersection \dots is the pair consisting of the observation model and the actual gossip state $(\initsecret,\sigma)$, that is valued with secret distribution $\initsecret[\sigma]$!

\subsection{Conclusion}

We determined when epistemic goals can be obtained for a simple gossip protocol given a bound of at most one transmission error. The dynamic epistemic analysis allows to determine and resolve conflicts in novel ways, that also result in more optimal executions wherein agents may learn the correct value of the secret of another agent without calling that agent. We wish to generalize our results to asynchrony, to $f$ transmission errors, and to $f$ faulty agents. Beyond that, we wish to generalize it to other distributed epistemic gossip protocols and to restricted network conditions, and to obtain hard bounds for various optimality questions.

\bibliographystyle{plain}
\bibliography{biblio2026_error}

@preamble{ {\providecommand{\noopsort}[1]{}} }

@inproceedings{Moura:2021:Lean4,
	author = {Moura, L. de and Ullrich, S.},
	title = {The {Lean} 4 {Theorem} {Prover} and {Programming} {Language}},
	doi = {10.1007/978-3-030-79876-5_37},
	booktitle = {Automated {Deduction} – {CADE} 28},
	publisher = {Springer},
	year = {2021},
	pages = {625--635},
}

@inproceedings{mathlib2020,
	author = {The mathlib Community},
	title = {The {Lean} mathematical library},
	doi = {10.1145/3372885.3373824},
	booktitle = {Proc.\ of 9th {CPP}},
	year = {2020},
	pages = {367--381},
}

@MastersThesis{line:2018,
  title     = {Unreliable gossip},
  author    = {L. van den Berg},
  school    = {University of Amsterdam},
  series    = {Master of Logic MoL-2018-01},
  year      = 2018,
  note = {MoL-2018-01},
  url       = {https://eprints.illc.uva.nl/9863/},
}

@MastersThesis{tinafurer:2023,
  title     = {Unreliability in Social Networks},
  author    = {T. Furer},
  school    = {University of Bern},
  year      = 2023,
}

@inproceedings{DaliotD05,
  author       = {A. Daliot and
                  D. Dolev},
  editor       = {T. Herman and
                  S. Tixeuil},
  title        = {Self-stabilization of Byzantine Protocols},
  booktitle    = {Proc.\ of 7th {SSS} (Self-Stabilizing Systems)},
  series       = {LNCS},
  volume       = {3764},
  pages        = {48--67},
  year         = {2005},
  doi          = {10.1007/11577327\_4},
}

@article{DolevFPSSL14,
  author       = {D. Dolev and
                  M. F{\"{u}}gger and
                  M. Posch and
                  U. Schmid and
                  A. Steininger and
                  C. Lenzen},
  title        = {Rigorously modeling self-stabilizing fault-tolerant circuits: An ultra-robust  clocking scheme for systems-on-chip},
  journal      = {J. Comput. Syst. Sci.},
  volume       = {80},
  number       = {4},
  pages        = {860--900},
  year         = {2014},
  doi          = {10.1016/J.JCSS.2014.01.001},
}

@book{dolev:2000,
author = {D. Dolev},
title = {Self-Stabilization},
publisher = {MIT Press}, year = {2000},
 }

@article{abs-1803-03042,
  author       = {A. Casta{\~{n}}eda and
                  J. Lef{\`{e}}vre and
                  A. Trehan},
  title        = {Self-healing Routing and Other Problems in Compact Memory},
  journal      = {CoRR},
  volume       = {abs/1803.03042},
  year         = {2018},
  url          = {http://arxiv.org/abs/1803.03042},
}

@InProceedings{DH08,
author="Dolev, D.
and Hoch, E.",
editor="Taubenfeld, G.",
title="Constant-Space Localized Byzantine Consensus",
booktitle="Distributed Computing",
year="2008",
pages="167--181",
}

@article{hvdetal.lucky:2024, 
title={You can only be lucky once: optimal gossip for epistemic goals}, 
DOI={10.1017/S0960129524000082}, 
journal={Mathematical Structures in Computer Science}, 
author={van Ditmarsch, H. and Gattinger, M.}, 
year={2024}, 
pages={1–28},
}

@unpublished{hvd:2026,
author= {H. van Ditmarsch}, 
title={Reasoning about Gossip},
note={Manuscript to appear with Cambridge University Press},
year = {2026},
}

@article{DitmarschGR23,
  author       = {H. van Ditmarsch and
                  M. Gattinger and
                  R. Ramezanian},
  title        = {Everyone Knows That Everyone Knows: Gossip Protocols for Super Experts},
  journal      = {Stud Logica},
  volume       = {111},
  number       = {3},
  pages        = {453--499},
  year         = {2023},
  doi          = {10.1007/S11225-022-10032-3},
}

@inproceedings{BergG20,
  author       = {L. van den Berg and
                  M. Gattinger},
  editor       = {M.A. Martins and
                  I. Sedl{\'{a}}r},
  title        = {Dealing with Unreliable Agents in Dynamic Gossip},
  booktitle    = {Proc.\ of 3rd DaL{\'{\i}}},
  note       = {{LNCS} 12569},
  pages        = {51--67},
  year         = {2020},
  doi          = {10.1007/978-3-030-65840-3\_4},
}

@article{AspnesR07,
  author       = {J. Aspnes and
                  E. Ruppert},
  title        = {An Introduction to Population Protocols},
  journal      = {Bull. {EATCS}},
  volume       = {93},
  pages        = {98--117},
  year         = {2007},
}

@article{west82a,
  author    = {D.B. West},
  title     = {A class of solutions to the gossip problem, part {I}},
  journal   = {Discrete Mathematics},
  volume    = {39},
  number    = {3},
  pages     = {307--326},
  year      = {1982},
}

@article{MosesT88,
  author    = {Y. Moses and
               M.R. Tuttle},
  title     = {Programming Simultaneous Actions Using Common Knowledge},
  journal   = {Algorithmica},
  volume    = {3},
  pages     = {121--169},
  year      = {1988},
  doi       = {10.1007/BF01762112},
}

@inproceedings{hvdetal.aiml:2022,
  author    = {H. van Ditmarsch and K. Fruzsa and R. Kuznets},
  title     = {A New Hope},
  booktitle = {Proc.\ of the 14th {AiML}},
  pages     = {349--369},
  year      = {2022},
  editor    = {D. Fern{\'{a}}ndez{-}Duque and A. Palmigiano},
  publisher = {College Publications},
}

@inproceedings{abs-2106-11499,
  author    = {K. Fruzsa and
               R. Kuznets and
               U. Schmid},
  title     = {Fire!},
  booktitle = {Proc.\ of the 18th {TARK}},
  pages     = {139--153},
  year      = {2021},
  doi       = {10.4204/EPTCS.335.13},
  editor    = {J.Y. Halpern and A. Perea},
  series    = {{EPTCS}},
  volume    = {335},
}

@inproceedings{KuznetsP0F19,
  author    = {R. Kuznets and
               L. Prosperi and
               U. Schmid and
               K. Fruzsa},
  title     = {Epistemic Reasoning with {B}yzantine-Faulty Agents},
  booktitle = {Proc.\ of 12th {FroCoS}},
  pages     = {259--276},
  year      = {2019},
  doi       = {10.1007/978-3-030-29007-8\_15},
  note    = {LNCS 11715},
}

@article{CooperHMMR19,
  author    = {M.C. Cooper and
               A. Herzig and
               F. Maffre and
               F. Maris and
               P. R{\'{e}}gnier},
  title     = {The epistemic gossip problem},
  journal   = {Discret. Math.},
  volume    = {342},
  number    = {3},
  pages     = {654--663},
  year      = {2019},
  doi       = {10.1016/j.disc.2018.10.041},
}

@article{logicofgossiping:2020,
author    = {H. van Ditmarsch and
               W. van der Hoek and
               L.B. Kuijer},
title = {The logic of gossiping},
journal   = {Artificial Intelligence},
  volume    = {286},
  pages     = {103306},
  year      = {2020},
  doi       = {10.1016/j.artint.2020.103306},
 }

@article{AptW18,
  author    = {K.R. Apt and
               D. Wojtczak},
  title     = {Verification of Distributed Epistemic Gossip Protocols},
  journal   = {J. Artif. Intell. Res.},
  volume    = {62},
  pages     = {101--132},
  year      = {2018},
  doi       = {10.1613/jair.1.11204},
}

@article{DitmarschGKP19,
  author    = {H. van Ditmarsch and
               M. Gattinger and
               L.B. Kuijer and
               P. Pardo},
  title     = {Strengthening Gossip Protocols using Protocol-Dependent Knowledge},
  journal   = {{FLAP}},
  volume    = {6},
  number    = {1},
  pages     = {157--203},
  year      = {2019},
}

@phdthesis{gattinger.phd:2018,
author = {M. Gattinger},
title = {New Directions in Model Checking Dynamic Epistemic Logic},
year = {2018},
school = {University of Amsterdam},
note = {ILLC Dissertation Series DS-2018-11},
 }

@article{hvdetal.dynamicgossip:2019,
  author    = {H. van Ditmarsch and
               J. van Eijck and
               P. Pardo and
               R. Ramezanian and
               F. Schwarzentruber},
title = {Dynamic Gossip},
volume = {45(3)},
journal = {Bulletin of the Iranian Mathematical Society},
year = {2019},
pages = {701--728},
doi = {10.1007/s41980-018-0160-4},
 }

@article{AgotnesW17,
  author    = {T. {\AA}gotnes and
               Y.N. W{\'{a}}ng},
  title     = {Resolving distributed knowledge},
  journal   = {Artif. Intell.},
  volume    = {252},
  pages     = {1--21},
  year      = {2017},
  doi = {10.1016/j.artint.2017.07.002},
}

@article{BAKER1972191,
title = "Gossips and telephones",
journal = "Discrete Mathematics",
volume = "2",
number = "3",
pages = "191 - 193",
year = "1972",
author = "B. Baker and R. Shostak"
}

@article{HerzigM17,
  author    = {A. Herzig and
               F. Maffre},
  title     = {How to share knowledge by gossiping},
  journal   = {{AI} Commun.},
  volume    = {30},
  number    = {1},
  pages     = {1--17},
  year      = {2017},
}

@article{DitmarschEPRS17,
  author    = {H. van Ditmarsch and
               J. van Eijck and
               P. Pardo and
               R. Ramezanian and
               F. Schwarzentruber},
  title     = {Epistemic protocols for dynamic gossip},
  journal   = {J. Applied Logic},
  volume    = {20},
  pages     = {1--31},
  doi       = {10.1016/j.jal.2016.12.001},
  year      = {2017},
}

@article{kermarrecetal:2007,
 author = {Kermarrec, A.-M. and van Steen, M.},
 title = {Gossiping in Distributed Systems},
 journal = {SIGOPS Oper. Syst. Rev.},
  volume = {41},
 number = {5},
 year = {2007},
 pages = {2--7},
 doi = {10.1145/1317379.1317381},
}

@proceedings{hvdetal.handbook:2015,
title = {Handbook of epistemic logic},
editor = {H. van Ditmarsch and J.Y. Halpern and W. van der Hoek and B. Kooi},
publisher = {College Publications},
year = {2015},
 }

@inproceedings{attamahetal.ecai:2014,
title = {Knowledge and Gossip},
author = {M. Attamah and H. van Ditmarsch and D. Grossi and W. van der Hoek},
booktitle = {Proc.\ of 21st {ECAI}},
pages = {21--26},
year = {2014},
publisher = {{IOS} Press},
  doi       = {10.3233/978-1-61499-419-0-21},
 }

@article{tijdeman:1971,
author = {R. Tijdeman}, 
title = {On a telephone problem}, 
journal = {Nieuw Archief voor Wiskunde},
volume = {3(19)}, 
pages = {188--192},
year = {1971},
 }

@article{hedetniemietal:1988,
author = {S.M. Hedetniemi and S.T. Hedetniemi and A.L. Liestman}, 
title = {A survey of gossiping and broadcasting in communication networks},
journal = {Networks}, 
volume = {18}, 
pages = {319--349},
year = {1988},
doi = {10.1002/net.3230180406},
 }

@book{halpern:2003,
author = {J.Y. Halpern},
title = {Reasoning about Uncertainty},
publisher = {MIT Press},
address = {Cambridge MA},
year = 2003  }

@misc{LeanAristotle2025,
title = {Aristotle: {IMO}-level Automated Theorem Proving},
url = {http://arxiv.org/abs/2510.01346},
shorttitle = {Aristotle},
author = {Achim, T. and Best, A. and {Bietti {et al.}}, A.},
year = {2025},
}


\section*{Appendix: formalizing self-correcting gossip in Lean}

Lean is a functional programming language and an interactive theorem prover~\cite{Moura:2021:Lean4}.
We formalized the semantics for self-correcting gossip described in this paper using Lean and its mathematical library \emph{mathlib}~\cite{mathlib2020}.
In particular we verify that the semantics is well-founded, and various properties of the notions of knowledge and belief.
Here we highlight some aspects of the formalization.
The full documentation of the formalization is available at
{\small \url{https://m4lvin.github.io/Gossip-in-Lean/docs/Gossip/Error/Basic.html}}

The whole formalization is (implicitly) parameterized by the number of agents \texttt{n}.
To represent the set of all agents we use \texttt{Fin n}, the finite type with \texttt{n} elements.
Calls are then defined as follows with three constructors, corresponding to the cases where there is no transmission error ($ab$), an error from the caller ($a^c b$), or an error from the callee ($a b^c$).
\begin{lstlisting}
inductive Call : Type
  /-- ⌜a b⌝ -/
  | normal : (caller : @Agent n)
           → (callee : { b : @Agent n // b ≠ caller }) → Call
  /-- ⌜a^c b⌝ -/
  | fstE : (caller : @Agent n)
         → (err : @Agent n)
         → (callee : { b : @Agent n // b ≠ caller }) → Call
  /-- ⌜a b^c⌝ -/
  | sndE : (caller : @Agent n)
         → (callee : { b : @Agent n // b ≠ caller })
         → (err : @Agent n) → Call
\end{lstlisting}

We make heavy use of the \texttt{notation} command in Lean so that we can denote agents, values and calls in a way that is more similar to our notation here in the paper.
For example, instead of \texttt{Call.sndE a b c} we can also write \lstinline{⌜a b^c⌝} for the call $a b^c$.

A sequence of calls is then just a \texttt{List} of calls.
For convenient pattern matching we let the head of the list denote the most recent call.
To also ensure sequences contain at most one error we define a function \lean{maxOne} and then let \lean{OSequence} be its \texttt{Subtype}.

A key part of the formalization are four functions that correspond to the definitions in section~\ref{section.syntaxsemantics}.
Below is the Lean code for \lean{eval} corresponding to $\models$ from Definition~\ref{def.semanticsformulas}.

\begin{lstlisting}
def eval : @Dist n → @OSequence n → @Form n → Prop
  | _, _, .Top        => True
  | S, σ, .Neg φ      => ¬ eval S σ φ
  | S, σ, .Has a (j,k) => (j,k) ∈ resultSet a S σ
  | S, σ, .Con φ ψ    => eval S σ φ ∧ eval S σ ψ
  | S, σ, .K a φ   => ∀ t, ∀ τ , (he : equiv a (S,⟨σ,rfl⟩) (t,τ)) → eval t τ φ
termination_by
  _ σ φ => (σ.length, φ.length) -- use lexicographic order to show termination
decreasing_by -- Sequence length stays the same, but formula becomes shorter.
  · apply Prod.Lex.right [...]
\end{lstlisting}
We need to convince Lean that the \lean{eval} function terminates, corresponding to the question whether our semantics is well-founded, as discussed on page~\pageref{par:SemWellDef}.
Just like the definitions in Section~\ref{section.syntaxsemantics} are mutually recursive, also in Lean we need to place the function \lean{eval} inside a \texttt{mutual} block together with functions that encode the other definitions:
\lean{resultSet} and \lean{contribSet} formalize Definition~\ref{def.semanticscall} for the semantics of calls,
and \lean{equiv} formalizes Definition~\ref{def.observationrelation} for the observation relation.
Each function then also has a \verb|termination_by| annotation that labels each recursive call with a value in the lexicographic order of pairs of sequence length and formula length.
In the \verb|decreasing_by| block we then provide proofs that indeed in each recursive call the values are decreasing.
We visualize the mutual recursion in Figure~\ref{fig:termination} where each arrow represents a recursive call from one function to another (or itself) and is labelled by how the lexicographic value changes.

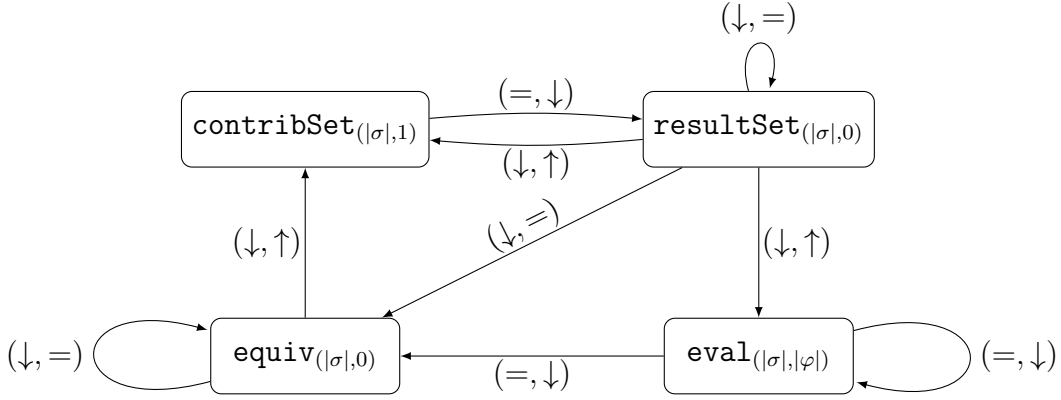
\begin{figure}
  \centering
  \begin{tikzpicture}[
    >=latex,
    boxnode/.style = {draw, rounded corners, minimum width=2.5cm, minimum height=1cm, align=center},
    edgelabel/.style = {midway, fill=none, draw=none, inner sep=1pt}
    ]

    \node[boxnode] (contrib) at (0, 3) {$\lean{contribSet}_{(|\sigma|,1)}$};
    \node[boxnode] (result)  at (6, 3) {$\lean{resultSet}_{(|\sigma|,0)}$};
    \node[boxnode] (equiv)   at (0, 0) {$\lean{equiv}_{(|\sigma|,0)}$};
    \node[boxnode] (eval)    at (6, 0) {$\lean{eval}_{(|\sigma|,|\phi|)}$};

    \draw[->] (contrib) edge [bend left=5] node[edgelabel, above] {$(=,\downarrow)$} (result);
    \draw[->] (result)  edge [bend left=5] node[edgelabel, below] {$(\downarrow,\uparrow)$} (contrib);

    \draw[->] (equiv)   -- node[edgelabel, left]  {$(\downarrow,\uparrow)$} (contrib);
    \draw[->] (result)  -- node[edgelabel, sloped, above] {$(\downarrow,=)$} (equiv);
    \draw[->] (result)  -- node[edgelabel, right] {$(\downarrow,\uparrow)$} (eval);

    \draw[->] (eval)    -- node[edgelabel, below] {$(=,\downarrow)$} (equiv);

    \draw[->] (result) edge[loop above] node {$(\downarrow,=)$} ();
    \draw[->] (equiv)  edge[loop left]  node {$(\downarrow,=)$} ();
    \draw[->] (eval)   edge[loop right] node {$(=,\downarrow)$} ();
  \end{tikzpicture}
  \caption{Mutual recursion and termination proof.
    Edge labels say how $(|\sigma|,|\phi|)$ changes.}\label{fig:termination}
\end{figure}

Again we define shorter notation, so that we can write \lstinline{S⌈σ⌉ ⊧ φ} instead of \lstinline{eval S σ φ}
and so that we can write \lstinline{S⌈σ⌉a} for \lstinline{resultSet a S σ}.
We also write \lstinline{b @ a} for the atom $b_a$ saying that $a$ has value $b$.
To denote values $b$ or $\overline{b}$ we use \lstinline{(b, k)} where \texttt{k} is a \texttt{Bool} value.

One of the first results we prove in Lean is \lean{equiv_Equivalence}, saying that $\sim_a$ is indeed an equivalence relation.
Besides this, we also show the results from Section~\ref{subsec:SemanticResults},
as shown in \autoref{fig:LeanLemmas}.
Parts of our Lean code were written by the \emph{Aristotle} tool~\cite{LeanAristotle2025}, and we refer to the git repository for details.

\begin{figure}[H]
\begin{lstlisting}
/-- Lemma #\ref{lemma.simtolocal}# -/
lemma #\lean{indistinguishable_then_same_values}# {n} {a : @Agent n} {S T: @Dist n}
  {σ τ : OSequence} : (S, σ) ~_a (T, τ)  →  S⌈σ⌉a = T⌈τ⌉a

/-- Lemma #\ref{lemma.localll}# -/
lemma #\lean{local_is_known}# {a b : @Agent n} (k : Bool) :
      ⊨ ((     ⟨b,k⟩ @ a ) ⟹ (K a (     ⟨b,k⟩ @ a) ))
    ∧ ⊨ ((Neg (⟨b,k⟩ @ a)) ⟹ (K a (Neg (⟨b,k⟩ @ a))))

/-- Lemma #\ref{lemma.stubborn}# -/
lemma #\lean{stubbornness}# m σ (h : σ.length = m) : S⌈σ⌉ ⊧ (a, k) @ a  ↔  S a = k

/-- Lemma #\ref{lemma.preservationofknowledge}# -/
lemma #\lean{knowledge_of_secrets_is_preserved}# {a b : @Agent n}
    (hKv : S⌈σ⌉ ⊧ Kv a b) (hSub : σ ⊑ τ) : S⌈τ⌉ ⊧ Kv a b

/-- Lemma #\ref{wieditleestisgek}# -/
lemma #\lean{consider_corrected}# (a b : @Agent n) {S : @Dist n} {σ : @OSequence n}
    {k : Bool} (real_b_is_k : S b = k) (a_has_no_b_k : (b, k) ∉ S⌈σ⌉a)
    : equiv a (S, ⟨σ, rfl⟩)
              (S.switch b, ⟨⟨cor b σ, cor_maxOne σ.2⟩, cor_same_length⟩)
    
/-- Proposition #\ref{cor.knowledgecorrectbelief}# -/
lemma #\lean{knowledge_implies_correct_belief}# {n} {a b : @Agent n} {k} :
  ⊨ (K a ((b,k) @ b)) ⟹ (((b,k) @ b) ⋀ ((b,k) @ a) ⋀ (¬'(b, !k) @ a))

/-- Corollary #\ref{cor.knowledgecorrectbelief2}# -/
lemma #\lean{knowledge_is_justified_true_belief}# {n} {a b : @Agent n} :
    ⊨ K a ((b,k) @ b) ⇔ K a ( ((b,k) @ b) ⋀ ((b,k) @ a) ⋀ (¬' (b, !k) @ a))
  \end{lstlisting}
  \caption{Overview of results proven in Lean (with links to documentation).}\label{fig:LeanLemmas}
\end{figure}

\end{document}